\documentclass[11pt]{article}
\usepackage[a4paper,margin=1in]{geometry}
\usepackage[utf8]{inputenc}
\usepackage[T1]{fontenc}

\usepackage{newtxtext,newtxmath}

\usepackage{graphicx}
\usepackage{setspace}
\usepackage{enumitem}

\usepackage{mathtools}

\usepackage{amsthm}

\usepackage{booktabs}
\usepackage{multirow}
\usepackage{adjustbox}
\usepackage{float}
\usepackage{makecell}
\usepackage{caption}
\usepackage{subcaption}

\usepackage{algorithm}
\usepackage[noend]{algpseudocode}
\algrenewcommand\algorithmicrequire{\textbf{Input:}}
\algrenewcommand\algorithmicensure{\textbf{Output:}}
\algrenewcommand\algorithmiccomment[1]{\hfill$\triangleright$~#1}
\algnewcommand{\LineComment}[1]{\State $\triangleright$~#1}
\usepackage[numbers]{natbib}
\usepackage[final,nopatch={footnote}]{microtype}

\usepackage[colorlinks=true, linkcolor=blue, citecolor=blue, urlcolor=blue]{hyperref}
\usepackage[nameinlink,noabbrev]{cleveref}

\newcounter{hypcount}

\newenvironment{hyp}{
  \begin{enumerate}[
    label=\textbf{\textup{(H\arabic*)}},
    ref=H\arabic*,
    leftmargin=*,
    align=left,
    itemsep=0.6ex,
    topsep=0.8ex
  ]
  \normalfont
  \setcounter{enumi}{\value{hypcount}}
}{
  \setcounter{hypcount}{\value{enumi}}
  \end{enumerate}
}

\newcommand{\Hyp}[1]{\textup{(\ref{#1})}}

\newenvironment{hypnamed}[1]{
  \begin{enumerate}[
    label=\textbf{\textup{(#1)}},
    ref={#1},
    leftmargin=*,
    align=left,
    itemsep=0.6ex,
    topsep=0.8ex
  ]
  \normalfont
}{
  \end{enumerate}
}

\theoremstyle{plain}
\newtheorem{theorem}{Theorem}[section]
\newtheorem{lemma}{Lemma}[section]
\newtheorem{proposition}[theorem]{Proposition}

\newtheoremstyle{dashremark}
  {6pt}
  {6pt}
  {\normalfont}
  {}
  {\bfseries\itshape}
  {.}
  {0.5em}
  {\thmname{#1}\thmnumber{~#2}\thmnote{ \textup{(#3)}}}

\theoremstyle{dashremark}
\newtheorem{remark}{Remark}[section]

\crefname{theorem}{theorem}{theorems}
\Crefname{theorem}{Theorem}{Theorems}
\crefname{lemma}{lemma}{lemmas}
\Crefname{lemma}{Lemma}{Lemmas}
\crefname{proposition}{proposition}{propositions}
\Crefname{proposition}{Proposition}{Propositions}
\crefname{remark}{remark}{remarks}
\Crefname{remark}{Remark}{Remarks}

\title{Stochastic Policy Gradient Methods in the Uncertain Volatility Model}
\author{
Lokman A. Abbas-Turki\thanks{Université Paris Cité and Sorbonne Université, Laboratoire de Probabilités, Statistique et Modélisation (LPSM, UMR CNRS 8001), Paris, France. \texttt{lokmane.abbas\_turki@sorbonne-universite.fr}} 
\and
Jean-François Chassagneux\thanks{ENSAE-CREST and Institut Polytechnique de Paris, Paris, France. \texttt{jean-francois.chassagneux@ensae.fr}. This research benefited from the support of the ``Chaire Futures of Quantitative Finance''.}
\and
Jean-Philippe Lemor\thanks{BNP Paribas Global Markets, Paris, France. \texttt{jean-philippe.lemor@bnpparibas.com}} 
\and Grégoire Loeper\thanks{BNP Paribas Global Markets, Paris, France. \texttt{gregoire.loeper@bnpparibas.com}} 
\and Simon Sananes\thanks{BNP Paribas Global Markets, Université Paris Cité and Sorbonne Université, Laboratoire de Probabilités, Statistique et Modélisation (LPSM, UMR CNRS 8001), Paris, France. \texttt{sananes@lpsm.paris}} 
}

\date{\today}

\begin{document}

\maketitle

\begin{abstract}

The multidimensional Uncertain Volatility Model leads to robust option 
pricing problems under joint volatility and correlation uncertainty. 
Their numerical resolution quickly becomes challenging because the 
associated stochastic control problem is high-dimensional. We propose 
a backward actor-critic stochastic policy-gradient scheme tailored to 
this setting. The method combines a discrete dynamic programming 
principle with Proximal Policy Optimization and one-hidden-layer 
neural-network approximations of both the value function and the 
control policy. A key ingredient is the policy parameterization: 
continuous controls are represented through a squashed Gaussian policy 
built on a $C$-vine representation of correlation matrices, which 
enforces positive definiteness by construction. Beyond the robust price itself, the spatial gradient of the trained critic provides an approximation of the associated superhedging strategy. We assess the quality of this learned gradient through a dual formulation, which also yields a numerical dual estimate of the price. Numerical experiments on a range of 
multidimensional derivatives show that the method yields accurate 
prices, remains computationally efficient, and compares favorably with 
existing Monte Carlo and machine-learning-based benchmarks for robust 
pricing in the Uncertain Volatility Model.

\end{abstract}


\section{Introduction}\label{sec:intro}

This paper develops a numerical method for high-dimensional robust 
pricing and hedging in the multidimensional Uncertain Volatility Model 
(UVM). In contrast to the classical Black-Scholes framework, where 
volatilities and correlations are specified once and for all, the UVM 
allows the instantaneous covariance structure of the risky assets to 
vary within a prescribed uncertainty set. The resulting value is the 
most conservative price compatible with this set and can be interpreted 
as an upper no-arbitrage price, or equivalently as a superreplication 
price under volatility and correlation uncertainty.

The one-dimensional UVM was introduced by~\cite{ALP95}. In that 
setting, the risky asset is modeled under a family of risk-neutral 
probability measures as
\[
dX_s = r X_s\,ds + \sigma_s X_s\,dW_s,\qquad s\in[0,T],
\]
where $r\ge 0$ is the risk-free rate and the progressively measurable 
volatility process $\sigma$ takes values in a prescribed interval 
$[\sigma_{\min},\sigma_{\max}]$. For a European payoff $g(X_T)$, 
the corresponding robust upper price is
\begin{equation}\label{eq:alp-v}
V_t \coloneqq \sup_{\mathbb{P}\in\mathcal P}
\mathbb{E}^{\mathbb{P}}_t
\left[\mathrm{e}^{-r(T-t)} g(X_T)\right],
\end{equation}
where $\mathcal P$ denotes the family of admissible martingale measures 
generated by such volatility processes, and 
$\mathbb{E}^{\mathbb{P}}_t$ denotes conditional expectation under 
$\mathbb{P}$ given the information available at time $t$. Under the 
usual assumptions, \cite{ALP95} show that this robust price coincides 
with the superreplication price: it is the smallest initial capital 
for which there exists a self-financing strategy whose terminal wealth 
dominates $g(X_T)$ under every admissible volatility scenario.

This framework naturally extends to multi-asset derivatives. In higher 
dimensions, uncertainty concerns the full instantaneous covariance 
matrix rather than a scalar volatility. In the formulation considered 
here, one imposes bounds on individual volatilities and correlations, 
while requiring the resulting correlation matrix to remain positive 
semidefinite. From 
a financial point of view, this setting is relevant when volatility and 
correlation are not directly traded and are only imperfectly observed. 
A rigorous stochastic control formulation of the multidimensional UVM 
is developed and analyzed in~\cite{GV02}.

The associated pricing problem can be written as a finite-horizon 
stochastic control problem in which the diffusion coefficient of the 
underlying assets is controlled. Its value function solves a parabolic, 
fully nonlinear Hamilton-Jacobi-Bellman (HJB) equation, where the 
nonlinearity involves a supremum over admissible covariance controls. 
Numerically solving this problem in high dimension is difficult for two 
related reasons: the state variable is high-dimensional, and the 
admissible control set is constrained by the positive semidefiniteness 
of the correlation matrix.

A first class of numerical methods addresses the fully nonlinear partial 
differential equation (PDE), or an equivalent second-order backward 
stochastic differential equation (2BSDE), directly. Building on the connection between 2BSDEs and fully 
nonlinear PDEs established in~\cite{CSTV07}, \cite{FTW11} introduce a 
probabilistic time-discretization scheme based on Monte Carlo 
approximations of conditional expectations. The approach was later 
adapted and refined in several directions; in particular, \cite{GH11} 
develop a 2BSDE-based methodology for the UVM, while \cite{GZZ15} relax 
some of the structural assumptions required for monotonicity. In these 
schemes, conditional expectations are typically approximated by 
regression methods, following the broad methodology of~\cite{GLW05}, 
and spatial derivatives may be estimated through Malliavin weights or 
Gaussian integration identities. Neural-network approximations have 
also been introduced to reduce the dependence on hand-crafted 
regression bases; see, e.g., \cite{PWG21} for backward schemes 
in which the value function and its gradient are parameterized by 
neural networks. Related neural PDE methods include the Deep BSDE 
method of~\cite{EHJ17, HJE18}, its extension to fully nonlinear 
equations in~\cite{BEJ19}, the Deep Galerkin method~\cite{SS18}, and 
physics-informed neural networks~\cite{RPK19}.

For the UVM, the HJB nonlinearity involves a maximization over 
admissible covariance matrices. When this maximization admits a closed 
form, PDE-2BSDE methods can be applied directly. In general, however, 
one must solve a nontrivial inner optimization problem at each time 
step and state point. This can be computationally expensive and may 
degrade accuracy, especially in high dimension. The numerical 
experiments of~\cite{GZZ15}, for example, illustrate the computational 
burden induced by the Hamiltonian maximization in a ten-dimensional 
uncertain volatility problem.

This difficulty motivates methods that exploit the stochastic control 
formulation directly, by parameterizing the control and optimizing over 
a finite-dimensional parameter space. In addition to their 
2BSDE-based scheme, \cite{GH11} propose such a control-parameterization 
strategy for the UVM. More recently, \cite{GMZ26} develop two numerical 
schemes specifically designed for high-dimensional UVM pricing. The first 
combines Gaussian-process regression with a multidimensional tree, while the 
second uses a neural-network parameterization of the optimal control. These 
methods provide some of the most extensive numerical benchmarks currently 
available for multidimensional UVM pricing, and constitute the main reference 
point for the experiments in~\Cref{sec:numerics}. More broadly, neural-network 
parameterizations of the control are used in global optimization approaches 
such as~\cite{HE16}, and differential learning techniques are incorporated 
in~\cite{LLP23} to improve the approximation of the value function and its 
derivatives. In contrast to global-in-time optimization strategies, 
\cite{HPBL21, BHLP22} develop backward methods based on dynamic 
programming, extending the parametric approach of~\cite{GH11} to more 
general stochastic control problems and higher-dimensional settings.

The stochastic control formulation also connects naturally with 
reinforcement learning. We refer to~\cite{SB18} for background on 
reinforcement learning, and to~\cite{WZZ20, JZ22a, JZ22b} for 
theoretical links between continuous-time stochastic control and 
reinforcement learning. Policy-gradient methods optimize parameterized 
controls by differentiating the expected performance criterion with 
respect to the policy parameters. Actor-critic methods combine this 
policy update with an approximation of the value function, reducing 
variance and stabilizing the optimization. Some algorithms in 
\cite{HPBL21, BHLP22} can be interpreted as actor-critic methods with 
deterministic controls parameterized by neural networks. A further 
possibility is to use stochastic policies, as in stochastic 
policy-gradient methods, which can improve exploration of the control 
space and reduce sensitivity to local optima. Such ideas have been 
used, for instance, by~\cite{HHP23} for stochastic control problems 
with exit times and by~\cite{PW25} for mean-field control problems.

Despite this progress, existing methods face specific limitations when 
applied to the multidimensional UVM. PDE-2BSDE methods require the 
solution of the Hamiltonian maximization unless the supremum can be 
simplified analytically. Global control-parameterization methods avoid 
this inner maximization but use a single set of parameters across all 
time steps, which may limit their flexibility in approximating 
time-dependent optimal controls. Backward dynamic programming methods 
address this issue, but existing implementations have mostly relied on 
deterministic control parameterizations. Moreover, the multidimensional 
UVM imposes a nontrivial structural constraint on the policy class: 
correlation controls must satisfy componentwise bounds while remaining 
compatible with positive semidefiniteness.

The present work addresses these issues by developing a backward 
actor-critic stochastic policy-gradient scheme for the multidimensional 
UVM. The method proceeds backward in time. At each time step, an actor 
network represents a stochastic policy for the volatility and 
correlation controls, while a critic network approximates the value 
function. The actor update uses the clipped Proximal Policy 
Optimization (PPO) surrogate objective~\cite{Sch+17}; in the present 
backward dynamic programming setting, this provides a stable stochastic 
approximation of the local policy-improvement step and permits several 
minibatch updates from the same simulated continuation values.

A central ingredient is the policy parameterization. Volatility 
controls are sampled through squashed Gaussian policies satisfying the 
prescribed bounds. Correlation controls are represented through 
$C$-vine coordinates~\cite{Joe06, JKL09}, again combined with smooth 
squashing maps. This yields positive definite correlation matrices by 
construction and avoids projections onto the cone of correlation matrices or repeated Cholesky feasibility 
checks during training. 

Obtaining a priori convergence rates or 
generalization bounds for neural-network-based stochastic control 
algorithms is a difficult problem. For backward dynamic programming 
schemes with deterministic control parameterizations, \cite{HPBL21} 
establish convergence in terms of neural-network approximation error 
and statistical estimation error, while leaving aside optimization 
error. Our setting introduces the additional complexity of stochastic 
policies and PPO updates, for which an analogous convergence theory is 
not currently available. The actor-critic scheme produces an actor-based Monte Carlo estimate which,
whenever the deployed deterministic policy is admissible, provides a lower
bound on the robust price, together with a pointwise approximation of the
value function through the critic. In addition, the spatial gradient of the critic provides an approximation of the associated superhedging strategy. This makes the quality of the learned gradient an important numerical object in its own right. We assess it through a dual construction based on stochastic-control duality.

Dual formulations for stochastic control problems replace the 
stochastic optimization problem by a family of deterministic, pathwise 
control problems penalized by a martingale correction. This idea was 
introduced in~\cite{DB92} and further developed in discrete time by 
\cite{Rog07}. Building on this framework, \cite{HLR16} derive a dual 
algorithm for diffusion control problems, including the UVM, and show 
that the dual bound converges to the true value as the control 
discretization is refined. A related recent contribution is~\cite{LY26}, 
which develops curse-of-dimensionality-free dual algorithms based on a 
stochastic PDE characterization and the generalized Hopf formula. In the present work, we use the diffusion-control dual approach of \cite{HLR16} as a diagnostic of the spatial gradient learned by the critic. The martingale correction in this construction requires an approximation of the gradient of the value function. We therefore introduce a Sobolev Training~\cite{Cza+17} variant of the critic, in which the network is trained to approximate both the value function and its gradient. The resulting gradient serves both as an approximate superhedging strategy and as the martingale integrand in the dual construction. The resulting dual estimate therefore provides a quantitative diagnostic of gradient quality, while also furnishing a numerical dual estimate of the robust price.

Our main contributions are:
\begin{itemize}
    \item We introduce a backward actor-critic stochastic 
    policy-gradient scheme for the multidimensional UVM. The method 
    combines discrete dynamic programming with a clipped PPO 
    policy-improvement step and neural-network approximations of the 
    value function and policy.

    \item We develop a stochastic policy parameterization for the 
    volatility and correlation controls. The correlation component is 
    based on a $C$-vine representation and a smooth squashing map, 
    yielding positive definite correlation matrices without projections or 
    repeated Cholesky checks. The prescribed pairwise correlation bounds 
    are incorporated through a penalty term in the actor objective.

    \item We train a Sobolev critic that approximates both the value function and its spatial gradient. The learned gradient provides an approximate superhedging strategy and is assessed through its use as the martingale integrand in a dual construction.

    \item We provide numerical evidence, across several multidimensional payoffs and dimensions, that the method yields stable price estimates and competitive runtimes relative to existing Monte Carlo and neural-network-based benchmarks. The dual experiments further show that Sobolev Training substantially improves the quality of the learned gradient.
\end{itemize}

The paper is organized as follows. \Cref{sec:UVM} recalls the 
mathematical framework of the UVM. \Cref{sec:algo} introduces the 
backward actor-critic scheme. \Cref{sec:numerics} reports the 
numerical experiments for the primal scheme. \Cref{sec:dual} studies the learned superhedging strategy through Sobolev Training and dual diagnostics.

\paragraph{Notation.}
Throughout the paper, $d\ge 1$ denotes the number of risky assets. For 
a matrix $M\in\mathbb{R}^{d\times d}$ and a vector 
$v\in\mathbb{R}^d$, we write $M^{ij}$ and $v^i$ for their components. 
The Hadamard product is denoted by $\odot$. We write 
$\operatorname{diag}(M)\in\mathbb{R}^d$ for the vector formed by the 
diagonal entries of $M$, and $\operatorname{diag}(v)$ for the diagonal 
matrix with diagonal vector $v$. The notation $|\cdot|$ denotes the Euclidean norm for vectors and the Frobenius norm for matrices.

We denote by $\mathbb{S}^d$ the set of $d\times d$ symmetric matrices, 
by $\mathbb{S}^d_+$ and $\mathbb{S}^d_{++}$ the cones of positive 
semidefinite and positive definite matrices, respectively, and by 
$\mathcal{E}^d$ the set of correlation matrices, namely the elements 
of $\mathbb{S}^d_+$ with unit diagonal. For a Borel set 
$E\subset\mathbb{R}^d$, $\mathcal{P}(E)$ denotes the set of probability 
measures on $E$, and $\mathcal{B}(E)$ its Borel $\sigma$-algebra.

The abbreviations a.e. and a.s. stand for almost everywhere and almost surely, respectively.

\section{Uncertain Volatility Model}\label{sec:UVM}

This section sets up the continuous- and discrete-time formulations of the
Uncertain Volatility Model used throughout the paper. We first place the
continuous-time model in the broader robust-pricing literature and clarify its
relation with the weak stochastic-control formulation of~\cite{GV02}. In the
present UVM, the choice of reference probability system does not affect the
value. This allows us to fix an arbitrary stochastic basis once and for all.
We then recall the associated Black-Scholes-Barenblatt equation, its
viscosity and verification characterizations, and the connection with
superreplication. Finally, we introduce the discrete-time UVM and its dynamic
programming recursion.

\subsection{Continuous-time Uncertain Volatility Model}\label{subsec:uvm-cont-time}

\subsubsection{Control formulation and independence of the stochastic basis}\label{subsubsec:uvm-control}

When the robust pricing problem in~\eqref{eq:alp-v} is represented on
canonical path space by a family of admissible pricing measures, this family is
in general nondominated and may contain mutually singular measures. This leads
to quasi-sure and nonlinear-expectation formulations of volatility uncertainty,
including the capacity approach of~\cite{DM06}, $G$-expectation~\cite{Pen08},
quasi-sure aggregation~\cite{STZ11}, second-order BSDEs~\cite{STZ12}, and
superhedging under volatility uncertainty~\cite{NS12}.

The present paper focuses on Markovian European multi-asset claims and follows
the stochastic-control formulation of~\cite{GV02}. The latter is itself
formulated weakly: the filtered probability space carrying the Brownian motion
is allowed to vary together with the admissible control. This additional
optimization over reference probability systems is logically distinct from the
nondominated-measure formulations just mentioned. The result below concerns
the value of the Markovian control problem of~\cite{GV02}. It does not identify
a general nondominated family of pricing measures with a single dominated
model.

Throughout the paper, $d\geq 1$ denotes the number of risky assets,
$T\in(0,+\infty)$ the maturity, and $r\in[0,+\infty)$ a constant
risk-free rate. The admissible instantaneous covariance matrices are
determined by two sets of bounds: bounds on individual volatilities,
\[
0<\sigma_{\min}^i\leq \sigma_{\max}^i<+\infty,
\qquad 1\leq i\leq d,
\]
and bounds on pairwise correlations,
\[
-1\leq \rho_{\min}^{ij}\leq \rho_{\max}^{ij}\leq 1,
\qquad 1\leq i<j\leq d.
\]
The associated set of admissible volatility coefficients is
\begin{equation*}
\begin{aligned}
A \coloneqq
\Big\{
a\in\mathbb R^{d\times d}:\;&
aa^\top
=
\operatorname{diag}(\sigma)\rho\operatorname{diag}(\sigma),\\
&
\sigma\in\prod_{i=1}^d[\sigma_{\min}^i,\sigma_{\max}^i],
\quad
\rho\in\mathcal E^d,\quad
\rho^{ij}\in[\rho_{\min}^{ij},\rho_{\max}^{ij}],
\quad \forall\,1\leq i<j\leq d
\Big\}.
\end{aligned}
\end{equation*}
We assume throughout that $A$ is nonempty. In dimensions $d=1$ and
$d=2$ this holds automatically. In dimensions $d\geq3$, however, the
pairwise correlation bounds may be mutually incompatible with the positive
semidefiniteness constraint defining $\mathcal E^d$. Under this assumption,
$A$ is compact. Indeed, the admissible covariance set is compact, so its
preimage under the continuous map $a\mapsto aa^\top$ is closed, while
$|a|^2=\operatorname{Tr}(aa^\top)=\sum_{i=1}^d(\sigma^i)^2$
is uniformly bounded on $A$. For any $a\in A$, the matrix $aa^\top$
represents an admissible instantaneous covariance matrix, and the
corresponding volatility vector and correlation matrix are recovered from
$aa^\top$ by
\[
\sigma^i=\sqrt{(aa^\top)^{ii}},
\qquad
\rho^{ij}
=
\frac{(aa^\top)^{ij}}{\sigma^i\sigma^j},
\qquad 1\leq i<j\leq d,
\]
which is well-defined since $\sigma_{\min}^i>0$ for all $i$. A given
admissible covariance matrix need not determine a unique factor $a$. At the
level of the HJB equation introduced below, only $aa^\top$ enters the
Hamiltonian. A specific parameterization of admissible factors is introduced
later for numerical purposes.

Let $g:(0,+\infty)^d\to\mathbb R$ be the payoff function of a European
claim with maturity $T$. We impose the following standing assumption.
\begin{hyp}
    \item\label{H1} $g$ is continuous and there exist constants
    $C_g>0$ and $m_g\geq0$ such that
    \[
        |g(x)|\leq C_g\bigl(1+|x|^{m_g}\bigr),
        \qquad x\in(0,+\infty)^d.
    \]
\end{hyp}
The payoffs considered in this paper satisfy~\Hyp{H1}.

Following the weak formulation of~\cite{GV02}, fix $t\in[0,T)$ and call a
reference probability system a tuple
\[
\nu_t
=
\bigl(\Omega,(\mathcal F_s)_{s\in[t,T]},\mathbb P,W\bigr),
\]
where $(\Omega,\mathcal F_T,\mathbb P)$ is a probability space,
$(\mathcal F_s)_{s\in[t,T]}$ is a filtration, and
$W=(W_s)_{s\in[t,T]}$ is a $d$-dimensional Brownian motion adapted to this
filtration. We write $\mathbb E^{\nu_t}$ for expectation under the probability
measure of $\nu_t$, and denote by $\mathcal A_{\nu_t}$ the set of all $A$-valued
processes on $[t,T]$ that are progressively measurable with respect to
$(\mathcal F_s)_{s\in[t,T]}$. For
$(t,x)\in[0,T)\times(0,+\infty)^d$ and
$\alpha\in\mathcal A_{\nu_t}$, let $X^{\alpha,t,x}$ solve
\begin{equation}\label{eq:UVM-multi-dynamics}
\begin{cases}
dX_s^{\alpha,t,x}
=
rX_s^{\alpha,t,x}\,ds
+
\operatorname{diag}\bigl(X_s^{\alpha,t,x}\bigr)\alpha_s\,dW_s,
& s\in[t,T],\\
X_t^{\alpha,t,x}=x.
\end{cases}
\end{equation}
Compactness of $A$ implies that the coefficients in
\eqref{eq:UVM-multi-dynamics} are Lipschitz with linear growth in $x$,
uniformly in $a\in A$. Consequently, standard results for controlled SDEs,
see \cite[Section~IV.2 and Appendix~D]{FS06}, yield a unique strong solution
for every reference probability system $\nu_t$ and every
$\alpha\in\mathcal A_{\nu_t}$. Moreover, for every $p\geq1$,
\begin{equation}\label{eq:moment-estimate}
\mathbb E^{\nu_t}\left[
\sup_{s\in[t,T]}\bigl|X_s^{\alpha,t,x}\bigr|^p
\right]
\leq C_p\bigl(1+|x|^p\bigr),
\end{equation}
where $C_p$ depends only on $p,T,r,$ and $A$, and is independent of
$\nu_t$ and $\alpha$. Since each component of
\eqref{eq:UVM-multi-dynamics} is a stochastic exponential times its
strictly positive initial value, a solution started from
$x\in(0,+\infty)^d$ remains in $(0,+\infty)^d$.

For a reference probability system $\nu_t$, define
\[
V_{\nu_t}(t,x)
\coloneqq
\sup_{\alpha\in\mathcal A_{\nu_t}}
\mathbb E^{\nu_t}\left[
\mathrm e^{-r(T-t)}g\bigl(X_T^{\alpha,t,x}\bigr)
\right].
\]
Under~\Hyp{H1} and~\eqref{eq:moment-estimate}, this value is finite and has
polynomial growth in $x$. Following~\cite[Eq.~(6)]{GV02}, the value of the weak formulation is obtained by taking the supremum over reference probability systems:
\begin{equation}\label{eq:uvm-weak-value}
V(t,x)\coloneqq\sup_{\nu_t}V_{\nu_t}(t,x),
\qquad (t,x)\in[0,T)\times(0,+\infty)^d,
\end{equation}
with $V(T,x)\coloneqq g(x)$.

For bounded uniformly continuous terminal rewards, reference-system
independence follows from Corollary~IV.7.1 of~\cite{FS06}, based on
Theorem~IV.7.1 therein, after a componentwise logarithmic transformation. The
polynomial-growth payoff in~\Hyp{H1} is handled by truncation.

\begin{proposition}[Independence of the stochastic basis]\label{prop:uvm-basis-independence}
Assume~\Hyp{H1}. Then, for every
$(t,x)\in[0,T)\times(0,+\infty)^d$ and every reference probability system
$\nu_t$,
\[
V(t,x)=V_{\nu_t}(t,x).
\]
\end{proposition}

An alternative route, which we do not pursue here, would be to combine the strong/weak comparison of~\cite[Section~4, Theorem~4.11]{ENJ87} with the weak dynamic programming principle of~\cite{BT11}, after the logarithmic transformation and payoff truncation used below.

The proof is deferred to Appendix~\ref{app:proofs}. After the componentwise
logarithmic transformation, the controlled drift has components
$r-\frac12(aa^\top)^{ii}$ and the diffusion coefficient is $a$. These
coefficients are bounded, continuous in the control variable, and independent
of time and state, so Corollary~IV.7.1 of~\cite{FS06} applies to bounded
uniformly continuous terminal rewards. Compactly supported truncations and the
moment estimate~\eqref{eq:moment-estimate} extend the conclusion to~\Hyp{H1}. No uniform ellipticity assumption is imposed on the original controlled
diffusion in this argument. The proof of~\cite[Theorem~IV.7.1]{FS06}
proceeds through uniformly parabolic approximations, but its conclusion
applies to the possibly degenerate setting considered here.

In view of Proposition~\ref{prop:uvm-basis-independence}, the outer optimization over
reference probability systems is redundant for the present UVM. We therefore
fix once and for all a filtered probability space
\[
(\Omega,\mathcal F,(\mathcal F_s)_{s\in[0,T]},\mathbb P)
\]
satisfying the usual conditions and carrying a $d$-dimensional
$(\mathcal F_s)_{s\in[0,T]}$-Brownian motion $W$. In particular, we do not assume that
$(\mathcal F_s)_{s\in[0,T]}$ is the augmented natural filtration of $W$. For
$t\in[0,T]$, let
\[
\mathcal A_t
\coloneqq
\Big\{
\alpha=(\alpha_s)_{s\in[t,T]}:
\alpha\text{ is progressively measurable with respect to }(\mathcal F_s)_{s\in[t,T]}\text{ and }A\text{-valued}
\Big\}.
\]
For $\alpha\in\mathcal A_t$, we keep the notation
$X^{\alpha,t,x}$ for the solution of~\eqref{eq:UVM-multi-dynamics} on this
fixed stochastic basis. Proposition~\ref{prop:uvm-basis-independence} gives the
representation
\begin{equation}\label{eq:uvm-value}
V(t,x)
=
\sup_{\alpha\in\mathcal A_t}
\mathbb E\left[
\mathrm e^{-r(T-t)}g\bigl(X_T^{\alpha,t,x}\bigr)
\right],
\qquad
(t,x)\in[0,T]\times(0,+\infty)^d.
\end{equation}
Thus all continuous-time control problems below may be formulated on this fixed
basis.

\subsubsection{Black-Scholes-Barenblatt equation, superreplication, and regularity}\label{subsubsec:bsb}

We first state the superreplication problem on the fixed stochastic basis above.
For a fixed $\alpha\in\mathcal A_0$, an $\mathbb R^d$-valued process $\Delta$
that is progressively measurable with respect to $(\mathcal F_s)_{s\in[0,T]}$
is called admissible if the stochastic gain
\[
\int_0^\cdot
\left\langle\Delta_s,d\bigl(\mathrm e^{-rs}X_s^{\alpha,0,x}\bigr)\right\rangle
\]
is well-defined and the associated discounted wealth process is a
supermartingale. Following the market formulation of~\cite{GV02}, define the
upper superreplication capital by
\begin{equation}\label{eq:superrep-capital}
\begin{aligned}
Y(x)\coloneqq\inf\Big\{y\in\mathbb R:\;&
\text{for every }\alpha\in\mathcal A_0\text{ there exists an admissible }
\Delta^{\alpha}\text{ such that}\\
&y+\int_0^T
\left\langle\Delta_s^{\alpha},d\bigl(\mathrm e^{-rs}X_s^{\alpha,0,x}\bigr)\right\rangle
\geq
\mathrm e^{-rT}g\bigl(X_T^{\alpha,0,x}\bigr)
\quad\mathbb P\text{-a.s.}\Big\}.
\end{aligned}
\end{equation}
Thus the trading strategy may depend on the volatility scenario. Under the
regularity conditions stated below, the same Markovian feedback rule
$\Delta_s^{\alpha}=\nabla_xV(s,X_s^{\alpha,0,x})$ superhedges all admissible
scenarios on the fixed basis.

The multidimensional Black-Scholes-Barenblatt (BSB) equation associated with
\eqref{eq:uvm-value} is
\begin{equation}\label{eq:BSB-multi}
\begin{cases}
\begin{aligned}
&\partial_tV(t,x)
+ r\left(x^\top\nabla_xV(t,x)-V(t,x)\right)\\
&\quad
+ \frac{1}{2}\sup_{a\in A}
G\big(x,\nabla^2_xV(t,x),a\big)
=0,
\end{aligned}
& (t,x)\in[0,T)\times(0,+\infty)^d,\\[3pt]
V(T,x)=g(x),
& x\in(0,+\infty)^d,
\end{cases}
\end{equation}
where, for $(x,\gamma,a)\in(0,+\infty)^d\times\mathbb{S}^d\times A$,
\begin{align*}
G(x,\gamma,a)
&\coloneqq
\operatorname{Tr}\left[
\gamma\,\operatorname{diag}(x)\,aa^\top\,\operatorname{diag}(x)
\right]\\
&=
\operatorname{Tr}\left[
\gamma\,\operatorname{diag}(x\odot\sigma)\,
\rho\,\operatorname{diag}(x\odot\sigma)
\right]\\
&=
\sum_{i,j=1}^d
x^i x^j \sigma^i\sigma^j\rho^{ij}\gamma^{ij},
\end{align*}
with $\sigma$ and $\rho$ denoting the volatility vector and correlation
matrix recovered from $aa^\top$.

By Proposition~\ref{prop:uvm-basis-independence}, the value function in
\eqref{eq:uvm-value} equals the weak value over reference probability systems.
After discounting, \cite[Theorem~7]{GV02} yields that, under~\Hyp{H1}, $V$ is
the unique viscosity solution of~\eqref{eq:BSB-multi} with polynomial growth.

Suppose now that
$V\in \mathcal{C}^{1,2}\big([0,T)\times(0,+\infty)^d;\mathbb R\big)$ and that
$V$, $\partial_tV$, $\nabla_xV$, and $\nabla_x^2V$ have polynomial growth. For
any $\alpha\in\mathcal A_0$, It\^o's formula and~\eqref{eq:BSB-multi} give
\begin{align}
 d\!\left(\mathrm e^{-rs}V(s,X_s^{\alpha,0,x})\right)
&=
\left\langle\nabla_xV(s,X_s^{\alpha,0,x}),
 d\bigl(\mathrm e^{-rs}X_s^{\alpha,0,x}\bigr)\right\rangle \notag\\
&\quad+
\frac12\mathrm e^{-rs}
\Big[
G\bigl(X_s^{\alpha,0,x},\nabla_x^2V(s,X_s^{\alpha,0,x}),\alpha_s\bigr)
\notag\\
&\hspace{42mm}-
\sup_{a\in A}G\bigl(X_s^{\alpha,0,x},\nabla_x^2V(s,X_s^{\alpha,0,x}),a\bigr)
\Big]ds.
\label{eq:discounted-ito-bsb}
\end{align}
The finite-variation term is nonpositive. The polynomial-growth assumptions and
\eqref{eq:moment-estimate} ensure the required integrability of the stochastic
integral, in particular, the gain process is a true martingale and the feedback
strategy is admissible. Hence
\[
V(0,x)+\int_0^T
\left\langle\nabla_xV(s,X_s^{\alpha,0,x}),
 d\bigl(\mathrm e^{-rs}X_s^{\alpha,0,x}\bigr)\right\rangle
\geq
\mathrm e^{-rT}g\bigl(X_T^{\alpha,0,x}\bigr)
\quad\mathbb P\text{-a.s.}
\]
for every $\alpha\in\mathcal A_0$. Thus $Y(x)\leq V(0,x)$. Conversely, if
$y$ is a superreplication capital, the supermartingale property in the
definition of admissibility implies
\[
y\geq
\mathbb E\!\left[\mathrm e^{-rT}g\bigl(X_T^{\alpha,0,x}\bigr)\right]
\qquad\text{for every }\alpha\in\mathcal A_0.
\]
Taking the supremum over $\alpha$ gives $y\geq V(0,x)$, and therefore
\begin{equation}\label{eq:superrep-value}
Y(x)=V(0,x).
\end{equation}
This is the discounted version of the smooth superreplication result in
\cite[Theorem~5]{GV02}.

Applying the same calculation from an arbitrary initial time $t$ gives the
control-theoretic characterization. Any admissible control $\alpha^\star$
satisfying
\begin{equation}\label{eq:alpha-star}
\alpha_s^\star\in
\arg\max_{a\in A}
G\Big(X_s^{\alpha^\star,t,x},
\nabla_x^2V\big(s,X_s^{\alpha^\star,t,x}\big),a\Big)
\end{equation}
for $ds\otimes d\mathbb P$-a.e. $(s,\omega)$ is optimal for
\eqref{eq:uvm-value}.

A sufficient set of conditions ensuring the required smoothness is the
following:
\begin{hyp}
\item \label{H2} the function $g$ is locally Lipschitz and $g$ and
$\nabla_x g$, defined a.e., have polynomial growth,
\item \label{H3} for every $a\in A$, the matrix $aa^\top$ is positive
definite.
\end{hyp}
Under~\Hyp{H2} and~\Hyp{H3}, compactness of $A$ implies that~\Hyp{H3} is
equivalent to a uniform strictly positive lower bound on the eigenvalues of
$aa^\top$ over $a\in A$. Moreover, \cite[Theorem~11]{GV02} implies that the
unique viscosity solution $V$ of~\eqref{eq:BSB-multi} belongs to the parabolic
H\"older space
$\mathcal{C}^{1+\beta/2,2+\beta}
\big([0,T)\times(0,+\infty)^d;\mathbb R\big)$ with polynomial growth for some
$\beta\in(0,1]$. In particular, the preceding It\^o argument applies and
\eqref{eq:superrep-value} holds.

The pointwise maximizing volatility~\eqref{eq:alpha-star} admits explicit
descriptions in low dimension, which we now record.

\paragraph{Dimension one.}
In the one-dimensional asset case
\[
dX_s^\alpha=rX_s^\alpha ds+\alpha_sX^\alpha_sdW_s,
\]
admissible covariance matrices $aa^\top$ lie in
$[(\sigma^1_{\min})^2,(\sigma^1_{\max})^2]$. For the purpose of evaluating the
Hamiltonian, we may choose the nonnegative representative
$a\in[\sigma^1_{\min},\sigma^1_{\max}]$. We have
$G(x,\gamma,a)=x^2a^2\gamma$, and the maximizing volatility level is explicit:
\begin{equation*}
\sigma^\star(\gamma)=
\sigma_{\min}^1\mathbf{1}_{\gamma<0}
+\sigma_{\max}^1\mathbf{1}_{\gamma\geq 0},
\end{equation*}
resulting in the one-dimensional BSB equation
\[
\begin{cases}
\begin{aligned}
\partial_tV(t,x)
&+r\left(x\partial_xV(t,x)-V(t,x)\right)\\
&+\tfrac{1}{2}x^2\left[
(\sigma^1_{\max})^2\left(\partial_{xx}^2V(t,x)\right)^+
-(\sigma^1_{\min})^2\left(\partial_{xx}^2V(t,x)\right)^-
\right]=0,
\end{aligned}
& (t,x)\in[0,T)\times(0,+\infty),\\
V(T,x)=g(x),
& x\in(0,+\infty).
\end{cases}
\]
In this case the pointwise maximizing volatility is bang-bang, i.e., valued in
$\{\sigma_{\min}^1,\sigma_{\max}^1\}$, and~\Hyp{H3} is equivalent to
$\sigma^1_{\min}>0$, which is already assumed throughout.

\paragraph{Dimension two.}
In dimension two, any correlation coefficient
\[
\rho^{12}\in[\rho_{\min}^{12},\rho_{\max}^{12}]\subset[-1,1]
\]
defines a valid correlation matrix
\[
\rho=
\begin{pmatrix}
1 & \rho^{12}\\
\rho^{12} & 1
\end{pmatrix}.
\]
Hence, for any admissible volatility vector $\sigma=(\sigma^1,\sigma^2)$, the
associated instantaneous covariance matrix is
\[
aa^\top=\operatorname{diag}(\sigma)\rho\operatorname{diag}(\sigma)
=
\begin{pmatrix}
(\sigma^1)^2 & \rho^{12}\sigma^1\sigma^2\\
\rho^{12}\sigma^1\sigma^2 & (\sigma^2)^2
\end{pmatrix}.
\]
For the purpose of evaluating the Hamiltonian, a convenient representative
$a\in A$ is the lower-triangular factor
\[
a=
\begin{pmatrix}
\sigma^1 & 0\\
\rho^{12}\sigma^2 & \sqrt{1-(\rho^{12})^2}\sigma^2
\end{pmatrix}.
\]
Therefore,
\[
G(x,\gamma,a)
=
(\sigma^1x^1)^2\,\gamma^{11}
+2\rho^{12}\,\sigma^1\sigma^2x^1x^2\,\gamma^{12}
+(\sigma^2x^2)^2\,\gamma^{22},
\]
and the pointwise maximizing correlation is bang-bang:
\begin{equation*}
\rho^{12,\star}
=
\rho_{\min}^{12}\mathbf{1}_{\gamma^{12}<0}
+\rho_{\max}^{12}\mathbf{1}_{\gamma^{12}\geq 0}.
\end{equation*}
In contrast, the pointwise maximizing volatilities are not bang-bang in
general: they depend on the Hessian through a nontrivial quadratic
optimization. Substituting $\rho^{12,\star}$ into $G(x,\gamma,a)$, the
maximization problem becomes the maximization of a two-dimensional quadratic
form under box constraints on the volatilities, whose explicit solution is
given in~\cite{GH11}. Since $\sigma^i_{\min}>0$, in dimension two~\Hyp{H3}
is equivalent to
$-1<\rho_{\min}^{12}\leq\rho_{\max}^{12}<1$.

\subsection{Discrete-time Uncertain Volatility Model}\label{subsec:discrete-UVM}

We now introduce a discrete-time approximation of the continuous-time
UVM of \Cref{subsec:uvm-cont-time}. This step is a natural prerequisite
for the numerical scheme developed in this paper. To this end, let
\begin{equation}\label{eq:time-discretization}
\{0\eqqcolon t_0<t_1<\dots<t_{N-1}<t_N\coloneqq T\},
\quad
t_n=n\tfrac{T}{N},
\end{equation}
denote a uniform time discretization of the interval $[0,T]$ with $N\geq 1$ time steps. We model our stochastic control problem as a discrete-time Markov control process with state process $\mathcal X$ taking values in $(0,+\infty)^d$, actions $a$ taken in the 
action set $A$, and transition law
\[
\begin{cases}
\mathcal{X}_0 \in (0,+\infty)^d,\\
\mathcal{X}_{n+1} = F(\mathcal{X}_n, a_n, \xi_n), 
\quad \forall\, 0 \leq n \leq N-1,
\end{cases}
\]
where $(\xi_n)_{0\leq n\leq N-1}$ is a sequence of i.i.d. random variables with distribution $\mathcal{N}(0,I_d)$ and 
$F$ is the log-Euler scheme defined by
\begin{equation}\label{eq:log-euler}
F(x,a,\xi) \coloneqq x \odot \exp\Big\{\big(r\mathbf{1}_d - \tfrac{1}{2}
\operatorname{diag}(aa^\top)\big)\tfrac{T}{N} 
+ a\sqrt{\tfrac{T}{N}}\,\xi\Big\},
\end{equation}
where $\mathbf{1}_d$ is the vector of $\mathbb{R}^d$ with all components equal to $1$.

The state $\mathcal{X}$ is controlled by means of stochastic Markov policies, which are sequences $\pi=(\pi_n)_{0\leq n\leq N-1}$ of stochastic kernels $\pi_n$ on $A$ given $(0,+\infty)^d$. That is, $\pi_n$ is a function such that $\pi_n(\cdot\mid x)$ is a probability measure on $A$ for each $x\in(0,+\infty)^d$, and $\pi_n(B\mid\cdot)$ is a measurable function on $(0,+\infty)^d$ for each $B\in\mathcal{B}(A)$. We write $\Pi$ for the set of stochastic kernels on $A$ given $(0,+\infty)^d$. Given $0\leq n\leq N-1$, we write $\Pi^{N-n}$ for the set of policies $\pi=(\pi_k)_{n\leq k\leq N-1}$ of length $N-n$, with $\Pi^1= \Pi$. We call $\pi\in\Pi^{N-n}$ a deterministic Markov policy when there exists a sequence of measurable $A$-valued functions $(\alpha_k)_{n\leq k\leq N-1}$ such that for all $x\in(0,+\infty)^d$ and $n\leq k\leq N-1$, $\pi_k(\cdot\mid x)$ is a Dirac measure concentrated at $\alpha_k(x)$. Finally, given a policy $\pi\in\Pi^{N-n}$, we will write $\mathcal{X}^{\pi,n,x}$ to denote the state $\mathcal{X}$ starting at time $t_n$ from the initial position $\mathcal{X}_{n}=x$, and controlled by $\pi$. For the sake of clarity, let us specify that the state $\mathcal{X}^{\pi,n,x}$ being controlled by $\pi$ means that at each step $k\geq n$, it evolves according to
\[
\mathcal{X}^{\pi,n,x}_{k+1} = F\big(\mathcal{X}^{\pi,n,x}_k, a_k, \xi_k\big),
\]
where the action $a_k\in A$ is drawn according to the probability measure $\pi_k\big(\cdot\mid \mathcal{X}^{\pi,n,x}_k\big)$.

We denote by $Q$ the transition kernel on
$(0,+\infty)^d$ defined by
\[
Q(E\mid x,a)
\coloneqq
\mathbb P\big(F(x,a,\xi)\in E\big),
\qquad
E\in\mathcal B((0,+\infty)^d),
\]
where $\xi\sim\mathcal N(0,I_d)$. Since $F$ is measurable and takes values
in $(0,+\infty)^d$, $Q$ is a stochastic kernel on $(0,+\infty)^d$ given
$(0,+\infty)^d\times A$.

Given an initial condition $(n,x)$ and a policy 
$\pi\in\Pi^{N-n}$, the Ionescu-Tulcea theorem 
(\cite[Appendix~C]{HLL96}) yields a unique probability measure 
$\mathbb P^{\pi,n,x}$ on the canonical path space
\[
\Omega^{N,n}\coloneqq
\big((0,+\infty)^d\big)^{N-n+1}\times A^{N-n},
\]
equipped with its product $\sigma$-algebra, under which 
$\mathcal X_n=x$ almost surely and, for every $n\leq k\leq N-1$,
\[
\mathbb P^{\pi,n,x}
\big(a_k\in C
\mid \mathcal X_n,a_n,\ldots,\mathcal X_{k-1},a_{k-1},\mathcal X_k\big)
=
\pi_k(C\mid \mathcal X_k),
\qquad C\in\mathcal B(A),
\]
and
\[
\mathbb P^{\pi,n,x}
\big(\mathcal X_{k+1}\in E
\mid \mathcal X_n,a_n,\ldots,\mathcal X_k,a_k\big)
=
Q(E\mid \mathcal X_k,a_k),
\qquad E\in\mathcal B((0,+\infty)^d).
\]
We write $\mathbb E^{\pi,n,x}$ for expectation under 
$\mathbb P^{\pi,n,x}$.

Our goal is to compute $\mathcal V_0^\star(\mathcal{X}_0)$, where the value functions $(\mathcal V_n^\star)_{0\leq n\leq N}$ are defined by
\begin{equation}\label{eq:Vn}
\begin{cases}
    \mathcal{V}_N^\star(x)\coloneqq g(x),\quad&\forall x\in(0,+\infty)^d,\\
    \mathcal{V}_n^\star(x)\coloneqq\sup_{\pi\in\Pi^{N-n}}\mathbb{E}^{\pi,n,x}\left[\mathrm{e}^{-r(T-t_n)}g\big(\mathcal{X}_N^{\pi,n,x}\big)\right],\quad&\forall x\in(0,+\infty)^d,\;\forall 0\leq n\leq N-1.
\end{cases}
\end{equation}

\Cref{lem:dDPP} below, whose proof is given in Appendix~\ref{app:proofs}, states the discrete-time dynamic programming principle (DPP) for the value functions~\eqref{eq:Vn}. This will be our main tool for the computation of $\mathcal{V}_0^\star(x)$. 

\begin{lemma}[Discrete DPP]\label{lem:dDPP}
Assume in addition that
\begin{hyp}
    \item \label{H4} the function $g$ is bounded below.
\end{hyp}
The sequence of value functions $(\mathcal{V}_n^\star)_{0\leq n\leq N}$ defined in~\eqref{eq:Vn} satisfies
\begin{equation}\label{eq:dDPP}
\mathcal{V}_n^\star(x) = \max_{a\in A}\mathbb{E}\left[\mathrm{e}^{-r\frac{T}{N}}\mathcal{V}_{n+1}^\star\big(F(x,a,\xi_n)\big)\right],\quad\forall x\in (0,+\infty)^d,\;\forall 0\leq n\leq N-1.
\end{equation}
Moreover, there exists a sequence $(\alpha_n^\star)_{0\leq n\leq N-1}$ of measurable functions such that the maximum in~\eqref{eq:dDPP} is attained at $\alpha^\star_n(x)$ for each $x\in (0,+\infty)^d$, i.e.,
\[
\mathcal{V}_n^\star(x) = \mathbb{E}\left[\mathrm{e}^{-r\frac{T}{N}}\mathcal{V}_{n+1}^\star\big(F(x,\alpha^\star_n(x),\xi_n)\big)\right],\quad\forall x\in (0,+\infty)^d,\;\forall 0\leq n\leq N-1.
\] 
\end{lemma}

We conclude this section with a quantitative convergence result showing that
$(\mathcal V_n^\star)_{0\leq n\leq N}$ approximates $V$ at the grid points.
For this convergence rate, we strengthen~\Hyp{H2} as follows:
\begin{hypnamed}{H2$'$}
    \item \label{H2prime} the function $g$ is globally Lipschitz on
    $(0,+\infty)^d$.
\end{hypnamed}

The proof of~\Cref{lem:convergence-rate} below is given in Appendix~\ref{app:proofs}.

\begin{lemma}\label{lem:convergence-rate}
Assume~\Hyp{H2prime} and~\Hyp{H4}. Then there exist constants $C>0$ and
$q\geq0$, independent of $N$, $n$, and $x$, such that for all
$N\geq1$, $0\leq n\leq N$, and $x\in(0,+\infty)^d$,
\begin{equation}\label{eq:cv-rate}
0
\leq
V(t_n,x)-\mathcal V_n^\star(x)
\leq
C(1+|x|^q)\big(\tfrac{T}{N}\big)^{1/4}.
\end{equation}
\end{lemma}

In particular,
\[
V(0,x)
\leq
\mathcal V_0^\star(x)
+
C(1+|x|^q)\big(\tfrac{T}{N}\big)^{1/4}.
\]

\section{Actor-critic stochastic policy-gradient scheme}\label{sec:algo}

Building on the discrete-time framework of~\Cref{subsec:discrete-UVM}, we now describe the numerical scheme used to compute robust prices in the UVM. The method combines a backward dynamic programming viewpoint with stochastic policy-gradient updates. At each time step, the policy (actor) and the value function (critic) are approximated by neural networks and trained alternately. We first recall the actor-critic and PPO framework underlying the method, then introduce the policy classes used in the paper, and finally describe the backward training procedure.

\subsection{Actor-critic framework}\label{subsec:AC-framework}

Under~\Hyp{H4}, the discrete DPP of~\Cref{lem:dDPP} has two main
implications. First, there exists an optimal Markov policy 
$\pi^\star=(\pi^\star_n)_{0\leq n\leq N-1}$ for problem~\eqref{eq:Vn}, 
and this policy is deterministic. Second, the value functions and 
optimal controls can be characterized by solving $N$ successive 
one-step maximization problems via~\eqref{eq:dDPP}, which we will 
exploit by training the actor and critic backward in time. The 
scheme described below relies on three ingredients.

The first ingredient consists of approximating the optimal policy $\pi^\star=(\pi^\star_n)_{0\leq n\leq N-1}$ by a family of stochastic Markov policies $(\pi_{\theta_n})_{0\leq n\leq N-1}$ represented by $N$ neural networks, each parameterized by $\theta_n\in\Theta$. We now focus on a single time step. In the remainder of this subsection, $n$ is fixed and $\theta$ denotes the actor parameter at step $n$. At each step $n=N-1,\dots,0$, the discrete DPP motivates the statewise performance criterion
\begin{equation}\label{eq:Jn}
\mathcal{J}_n(x;\theta)\coloneqq \mathbb{E}^{a\sim\pi_{\theta}(\cdot\mid x)}\left[\mathrm{e}^{-r\frac{T}{N}}\mathcal{V}^\star_{n+1}\big(F(x,a,\xi_n)\big)\right],\qquad \forall x\in(0,+\infty)^d,
\end{equation}
where the expectation is taken jointly over 
$a\sim\pi_\theta(\cdot\mid x)$ and $\xi_n\sim\mathcal{N}(0,I_d)$, independently. For each fixed state $x$, one would like to solve
\begin{equation*}
\max_{\theta\in\Theta}\mathcal{J}_n(x;\theta).
\end{equation*}
However, because the same neural network parameter $\theta$ must act simultaneously on many states, the actual training problem is not pointwise in $x$. In practice, the actor is therefore trained through an aggregated objective $\mathbb{E}^{x\sim\mu_n}\big[\mathcal{J}_n(x;\theta)\big]$, where $\mu_n$ is a state-sampling distribution on $(0,+\infty)^d$ introduced in~\Cref{subsec:training}. We derive the gradient identities below at the statewise level $\mathcal{J}_n(x;\theta)$; the corresponding identities for the aggregated objective follow by averaging over $\mu_n$.

Although the optimal policy is deterministic, we parameterize a family of stochastic policies. This standard practice in reinforcement learning provides a smooth parameterization that admits likelihood-ratio policy-gradient estimators~\cite{Wil92}; see~\cite{Sil+14} for a discussion of stochastic versus deterministic policy gradients, and~\cite{Mon+24} for an analysis of deterministic-policy learning through stochastic policy classes. 

The second ingredient is the use of stochastic policy-gradient methods to solve the step-$n$ optimization problem. The general heuristic of policy-gradient methods is to approximate a maximizer $\theta_n^\star$ by gradient ascent, given an estimator of the gradient of the performance criterion. Score-function gradient estimators require the policies to admit densities. We therefore assume
\begin{hyp}
    \item \label{H5} for each $0\leq n\leq N-1$, there exists a measure $\kappa$ on the corresponding action space such that
    every stochastic policy in $\{\pi_{\theta_n},\theta_n\in\Theta\}$ admits a
    density with respect to $\kappa$, i.e., there exist parameterized measurable functions $\{p_{\theta_n},\theta_n\in\Theta\}$ such that
    \[
    \pi_{\theta_n}(da\mid x)=p_{\theta_n}(a\mid x)\kappa(da).
    \]
\end{hyp}
This assumption is satisfied by the policy classes introduced
in~\Cref{subsec:policies}. Assuming in addition the standard differentiability
and integrability conditions that justify differentiation under the integral
sign, the score-function identity underlying the REINFORCE estimator
of~\cite{Wil92} reads
\begin{equation}\label{eq:reinforce}
\nabla_{\theta}\,\mathcal{J}_n(x;\theta)
=\mathbb{E}^{a\sim\pi_{\theta}(\cdot\mid x)}\left[\mathrm{e}^{-r\frac{T}{N}}\mathcal{V}^\star_{n+1}\big(F(x,a,\xi_n)\big)\,\nabla_{\theta}\log p_{\theta}(a\mid x)\right],
\end{equation}
which is proved in Appendix~\ref{app:proofs}. The gradient~\eqref{eq:reinforce} is then estimated by Monte Carlo using i.i.d. samples of $a$ and $\xi_n$. In practice, we do not have access to the exact value function $\mathcal{V}_{n+1}^\star$. We therefore represent the value functions $(\mathcal{V}_{n}^\star)_{0\leq n\leq N-1}$ by $N$ neural networks $(\mathcal{V}_{\phi_n})_{0\leq n\leq N-1}$, each parameterized by $\phi_n\in\Phi$. Replacing in~\eqref{eq:reinforce} the unknown function $\mathcal{V}_{n+1}^\star$ by the already learned approximation $\mathcal{V}_{\phi_{n+1}^\star}$ yields the practical gradient surrogate
\[
\nabla_{\theta}\,\mathcal{J}_n(x;\theta)\approx
\mathbb{E}^{a\sim\pi_{\theta}(\cdot\mid x)}\left[
\mathrm{e}^{-r\frac{T}{N}}\mathcal{V}_{\phi_{n+1}^\star}\big(F(x,a,\xi_n)\big)\,\nabla_{\theta}\log p_{\theta}(a\mid x)
\right].
\]

Maintaining at each step $n$ two neural networks, one for the policy approximation and one for the value function approximation, makes our method belong to the actor-critic class. Actor-critic methods alternate between two coupled learning steps. The critic refers to the neural network $\mathcal{V}_{\phi_n}$, which learns the value under the current policy, typically by regression. The actor refers to the neural network $\pi_{\theta_n}$, which learns the policy by stochastic policy-gradient methods. The critic and actor are coupled through the value function under the current policy: the critic learns this value, and the actor uses it as a control variate to reduce gradient variance. Indeed, the same arguments used to prove~\eqref{eq:reinforce} yield for any measurable function $b:(0,+\infty)^d\to\mathbb{R}$,
\[
\mathbb{E}^{a\sim\pi_{\theta}(\cdot\mid x)}\left[b(x)\nabla_\theta\log p_{\theta}(a\mid x)\right]=0,
\]
so we may replace~\eqref{eq:reinforce} by
\[
\nabla_{\theta}\,\mathcal{J}_n(x;\theta)\approx
\mathbb{E}^{a\sim\pi_{\theta}(\cdot\mid x)}\left[
\Big(\mathrm{e}^{-r\frac{T}{N}}\mathcal{V}_{\phi_{n+1}^\star}\big(F(x,a,\xi_n)\big)-b(x)\Big)\,\nabla_{\theta}\log p_{\theta}(a\mid x)
\right].
\]
A standard and effective choice is to take $b$ as the value function under the current policy (see, e.g., \cite{GBB04}), which is what the critic $\mathcal{V}_{\phi_n}$ is designed to approximate. This leads to the practical gradient surrogate
\begin{equation}\label{eq:reinforce-approx}
\nabla_{\theta}\,\mathcal{J}_n(x;\theta)\approx
\mathbb{E}^{a\sim\pi_{\theta}(\cdot\mid x)}\left[
\Big(\mathrm{e}^{-r\frac{T}{N}}\mathcal{V}_{\phi_{n+1}^\star}\big(F(x,a,\xi_n)\big)-\mathcal{V}_{\phi_n}(x)\Big)\,\nabla_{\theta}\log p_{\theta}(a\mid x)
\right].
\end{equation}

In practice, these gradient identities are used to perform gradient ascent
steps in parameter space, with learning rate $\eta>0$. At a generic
optimization iteration $k$, this suggests the schematic update
\[
\theta_n^{(k+1)}
\leftarrow
\theta_n^{(k)}
+\eta\,\nabla_{\theta}\,
\mathbb{E}^{x\sim\mu_n}\big[\mathcal{J}_n(x;\theta)\big]
\big|_{\theta=\theta_n^{(k)}},
\]
where the expectation over states is approximated empirically using sampled states drawn from $\mu_n$. However, it has long been observed that in policy space, a small step in parameter space may induce a large change in the distribution $\pi_{\theta}(\cdot\mid x)$, which can significantly degrade performance, especially when policies are represented by neural networks and the action space is high-dimensional.

The third ingredient of our methodology is therefore an update rule that
controls the magnitude of policy changes. We first reformulate the
performance criterion in a way that allows reusing samples generated
under the current policy. Let $\theta_n^{(\ell)}$ denote the current
parameters and let $\theta$ be a candidate update. The likelihood-ratio
identity yields
\[
\mathcal{J}_n(x;\theta)
=
\mathbb{E}^{a\sim\pi_{\theta_n^{(\ell)}}(\cdot\mid x)}\left[
\frac{p_{\theta}(a\mid x)}{p_{\theta_n^{(\ell)}}(a\mid x)}\,
\mathrm{e}^{-r\frac{T}{N}}\,
\mathcal{V}_{n+1}^\star\big(F(x,a,\xi_n)\big)
\right].
\]
Such importance-sampling reformulations allow samples generated under the
current policy to be reused when evaluating nearby candidate policies: the candidate $\theta$ is evaluated using
samples generated under $\theta_n^{(\ell)}$, without re-sampling for
each trial value of $\theta$.

Building on this reformulation, Trust Region Policy Optimization (TRPO)
of~\cite{Sch+15} maximizes the resulting surrogate objective subject
to an explicit constraint on the change in policy. At the statewise
level, this may be written schematically as
\[
\theta_n^{(\ell+1)}
\in
\arg\max_{\theta\in\Theta}\mathcal{J}_n(x;\theta)
\quad\text{s.t.}\quad
\operatorname{KL}\big(\pi_{\theta_n^{(\ell)}}(\cdot\mid x)\,\|\,\pi_{\theta}(\cdot\mid x)\big)\leq \tau,
\]
where $\tau>0$ controls the maximal allowed deviation between successive policies and $\operatorname{KL}$ denotes the Kullback-Leibler divergence. In practice, TRPO is computationally demanding:
it relies on second-order information through the Hessian of the
Kullback-Leibler divergence, combined with a conjugate gradient
algorithm and a backtracking line search. We therefore adopt the Proximal Policy Optimization (PPO) method of~\cite{Sch+17} instead.

PPO replaces the constrained optimization problem by a simpler clipped
surrogate objective, in which a clipped version of the likelihood ratio
\[
\frac{p_{\theta}(a\mid x)}{p_{\theta_n^{(\ell)}}(a\mid x)}
\]
is formed using
\[
\operatorname{clip}(y,1-\varepsilon,1+\varepsilon)\coloneqq \min\big\{\max\{y,1-\varepsilon\},\,1+\varepsilon\big\}.
\]
This clipping prevents excessively large policy updates while retaining the simplicity of first-order gradient methods. Combining this idea with the approximation~\eqref{eq:reinforce-approx}, the resulting statewise PPO objective reads
\begin{align}
\mathcal{J}_n^{\mathrm{PPO}}(x,\theta_n^{(\ell)};\theta)
\coloneqq
\mathbb{E}^{a\sim\pi_{\theta_n^{(\ell)}}(\cdot\mid x)}\Bigg[
\min\Bigg\{&
\frac{p_{\theta}(a\mid x)}{p_{\theta_n^{(\ell)}}(a\mid x)}\,\operatorname{Adv}_n(x,a,\xi_n),\nonumber\\
&
\operatorname{clip}\Bigg(
\frac{p_{\theta}(a\mid x)}{p_{\theta_n^{(\ell)}}(a\mid x)},
1-\varepsilon,1+\varepsilon
\Bigg)\operatorname{Adv}_n(x,a,\xi_n)
\Bigg\}
\Bigg]\label{eq:PPO-objective},
\end{align}
where
\begin{equation}\label{eq:adv}
\operatorname{Adv}_n(x,a,\xi_n)\coloneqq
\mathrm{e}^{-r\frac{T}{N}}\mathcal{V}_{\phi_{n+1}^\star}\big(F(x,a,\xi_n)\big)-\mathcal{V}_{\phi_n}(x)
\end{equation}
is the corresponding one-step sampled advantage estimate. As for the performance criterion in~\eqref{eq:Jn}, the practical PPO actor objective is obtained from~\eqref{eq:PPO-objective} by aggregating over the state-sampling distribution $\mu_n$ and replacing the expectations over $(x,a,\xi_n)$ by empirical minibatch averages. Concrete sampling and minibatch construction are described in~\Cref{subsec:training}.

\subsection{Policy parameterization}\label{subsec:policies}

The actor-critic framework of~\Cref{subsec:AC-framework} is written on the
admissible action set $A$. For the numerical implementation, we use a
penalized relaxation in volatility-correlation coordinates. The construction
relies on a $C$-vine parameterization of correlation matrices, illustrated
in~\eqref{eq:c-vine-d3}, which enforces positive definiteness by construction.
For $d\geq3$, the prescribed pairwise correlation bounds are not imposed
structurally and are instead handled through the penalty introduced
in~\Cref{subsubsec:correlation-penalty}. The cases $d=1$ and $d=2$ are simpler
and are treated directly as described there.

In the UVM, the natural control variables are the volatilities and
correlations, rather than a diffusion matrix representation. We therefore
construct policies directly on volatility-correlation coordinates. A
sampled control is parameterized by a pair $(\sigma,\rho)$ where $\sigma=(\sigma^1,\dots,\sigma^d)$ is a volatility vector and
$\rho\in\mathcal E^d\cap\mathbb{S}^d_{++}$ is a positive definite
correlation matrix. Whenever these coordinates are used in the transition
map $F(x,a,\xi)$ introduced in~\Cref{subsec:discrete-UVM}, we select the
lower-triangular Cholesky representative
\[
a=\operatorname{diag}(\sigma)L(\rho),
\qquad L(\rho)L(\rho)^\top=\rho.
\]
Thus, throughout the numerical part of the paper, $(\sigma,\rho)$ denotes
the volatility-correlation coordinates of a control, whereas $a$ denotes
the corresponding diffusion matrix. In particular, the policy densities
below are written with respect to the coordinates $(\sigma,\rho)$.

If $\sigma_{\min}^i=\sigma_{\max}^i$ for some $i$, the corresponding
volatility coordinate is fixed and can simply be omitted from the policy
parameterization. To simplify notation, we describe below the case of
nondegenerate volatility intervals. Accordingly, we consider the support
\begin{equation}\label{eq:set-B}
B\coloneqq
\Big\{
(\sigma,\rho)\in\mathbb R^d\times\big(\mathcal E^d\cap\mathbb{S}^d_{++}\big)\,:\,
\sigma^i\in\big(\sigma_{\min}^i,\sigma_{\max}^i\big),\ \forall 1\leq i\leq d
\Big\},
\end{equation}
and for each state $x\in(0,+\infty)^d$, the policy is a probability measure 
\[
\pi_\theta(\cdot\mid x)\in\mathcal P(B).
\]
The set $B$ is a relaxed coordinate action space relative to the admissible
set of the discrete DPP. It enforces strict positive definiteness of $\rho$
and the volatility constraints through the parameterization, while, for
$d\geq3$, the prescribed pairwise correlation bounds are left to the penalty
term. This relaxation is used because the $C$-vine parameterization introduced
below yields valid correlation matrices but does not respect arbitrary box
constraints on their individual off-diagonal entries.

A direct Gaussian policy on $B$ is not well suited, since $B$ is a bounded action set with a nonlinear structural constraint on the correlation component. We therefore follow the standard idea of squashed Gaussian policies in reinforcement learning, as popularized by Soft Actor-Critic methods~\cite{Haa+18}: one starts from an unconstrained Gaussian latent variable and maps it to the action domain through a $\tanh$-squashing transformation. Starting from an unconstrained latent variable $z\in\mathbb{R}^{\frac{d(d+1)}{2}}$, we first apply $\tanh$ elementwise to obtain a bounded variable $\tanh(z)\in(-1,1)^{\frac{d(d+1)}{2}}$. 

Splitting $z$ as $z=(z_\sigma,z_\rho)\in\mathbb{R}^d\times\mathbb{R}^{\frac{d(d-1)}{2}}$, the first component $z_\sigma$ is then mapped to the admissible volatility intervals by a rescaling and shifting:
\[
\sigma^i \coloneqq \frac{\sigma^i_{\max}+\sigma^i_{\min}}{2}
+\frac{\sigma^i_{\max}-\sigma^i_{\min}}{2}\tanh(z_\sigma^i),
\qquad \forall 1\leq i\leq d.
\]
The correlation component $z_\rho$ requires a more structured 
construction to account for the positive semidefiniteness 
constraint. After the $\tanh$-squashing, $\tanh(z_\rho)$ belongs 
to $(-1,1)^{d(d-1)/2}$ and its entries are interpreted as partial 
correlations. Recall that, for distinct indices $i,j$ and a conditioning
set $K$, the partial correlation $\rho^{ij\mid K}$ is the correlation
between the residuals obtained by linearly projecting the $i$-th and
$j$-th components onto the components indexed by $K$. In particular,
$\rho^{ij\mid\varnothing}=\rho^{ij}$. Partial correlations can be
computed recursively as
\[
\rho^{ij\mid K\cup\{k\}}
=
\frac{
\rho^{ij\mid K}
-\rho^{ik\mid K}\rho^{jk\mid K}
}{
\sqrt{\bigl(1-(\rho^{ik\mid K})^2\bigr)
       \bigl(1-(\rho^{jk\mid K})^2\bigr)}
},
\]
for $k\notin K\cup\{i,j\}$. We then rely on the $C$-vine parameterization
of~\cite{Joe06,JKL09}, which provides a smooth bijection
\[
C:(-1,1)^{\frac{d(d-1)}{2}}\ni y\mapsto C(y)\in\mathcal{E}^d\cap\mathbb{S}^d_{++}
\]
between partial-correlation vectors and positive definite correlation matrices, see~\cite{JKL09}. This generates positive 
definite correlation matrices directly from unconstrained latent variables, without requiring projection onto the cone of correlation matrices or rejection of invalid samples. This differs from the NNU approach of~\cite{GMZ26}, where the neural network 
outputs tentative pairwise correlations and positive semidefiniteness is enforced 
by a safeguard: if the resulting correlation matrix is not positive semidefinite, 
it is replaced in the forward simulation by a fixed admissible matrix, together 
with an external penalty discouraging infeasible outputs. In contrast, the 
$C$-vine parameterization used here enforces positive definiteness structurally; 
no projection, rejection, or fallback correlation matrix is needed during 
simulation. Concretely, when $d=3$, the ordinary pairwise correlations are recovered
from the $C$-vine partial correlations
$(y_{12},y_{13},y_{23\mid 1})$ through
\begin{equation}\label{eq:c-vine-d3}
\rho^{12}=y_{12},\qquad
\rho^{13}=y_{13},\qquad
\rho^{23}=y_{23\mid 1}\sqrt{(1-(\rho^{12})^2)(1-(\rho^{13})^2)}+\rho^{12}\rho^{13}.
\end{equation}
For $d\geq 4$, the map is defined recursively; we refer to~\cite{JKL09} 
for explicit formulas.

Composing the $\tanh$-squashing with the rescaling and shifting map for the volatility component, and the $C$ map for the correlation component, we obtain a smooth bijection
\[
\mathcal{T}_{\mathrm{UVM}}:\mathbb{R}^{\frac{d(d+1)}{2}}\ni z\mapsto \mathcal{T}_{\mathrm{UVM}}(z)\in B.
\]
The resulting policy is obtained by representing the mean of the latent Gaussian by a neural network
\[
m_\theta:(0,+\infty)^d\ni x\mapsto m_{\theta}(x)\in \mathbb{R}^{\frac{d(d+1)}{2}},
\]
with parameter $\theta$. At training epoch $\ell$, we denote by
$\Lambda^{(\ell)}$ the positive definite diagonal covariance matrix of the
latent Gaussian. The covariance is fixed within each epoch but varies across
epochs according to the annealing schedule described in
\Cref{subsubsec:exploration-exploitation}. Whenever the epoch is fixed, we
write $\Lambda=\Lambda^{(\ell)}$ to lighten the notation. For each state $x\in(0,+\infty)^d$, we define the policy $\pi_\theta(\cdot\mid x)$ as the pushforward of the Gaussian law $\mathcal N(m_\theta(x),\Lambda)$ through $\mathcal T_{\mathrm{UVM}}$, that is,
\[
\pi_\theta(\cdot\mid x)
=
(\mathcal T_{\mathrm{UVM}})_{\#}\,\mathcal N(m_\theta(x),\Lambda).
\]
Let $\operatorname{Leb}$ denote Lebesgue measure on the latent space and set
\[
\kappa\coloneqq(\mathcal T_{\mathrm{UVM}})_\# \operatorname{Leb}.
\]
Since $\mathcal T_{\mathrm{UVM}}$ is bijective and does not depend on
$\theta$, the policy $\pi_\theta(\cdot\mid x)$ admits, with respect to
$\kappa$, the density
\[
p_\theta(\sigma,\rho\mid x)
=
\varphi\big(
\mathcal T_{\mathrm{UVM}}^{-1}(\sigma,\rho);
m_\theta(x),\Lambda
\big),
\qquad (\sigma,\rho)\in B,
\]
where $\varphi(\cdot\,;m,\Lambda)$ denotes the density of
$\mathcal N(m,\Lambda)$.

A practical benefit of the $C$-vine construction is that the 
Cholesky factor $L$ of the correlation matrix can be read off 
directly from the partial correlations, without requiring a 
numerical factorization. Each row $i$ of the lower-triangular 
factor $L$ satisfying $\rho=LL^\top$ is built from the partial 
correlations involving asset $i$ at successive vine levels. For 
$d=3$,
\[
L = \begin{pmatrix}
1 & 0 & 0\\
y_{12} & \sqrt{1-y_{12}^2} & 0\\
y_{13} & y_{23\mid 1}\sqrt{1-y_{13}^2} 
& \sqrt{(1-y_{13}^2)(1-y_{23\mid 1}^2)}
\end{pmatrix}.
\]
Since each finite latent value is mapped by $\tanh$ into $(-1,1)$, each factor 
$\sqrt{1-y^2}$ is strictly positive and $L$ has positive diagonal. This is exploited in the 
log-Euler scheme~\eqref{eq:log-euler}: the associated diffusion matrix $\operatorname{diag}(\sigma)L$ is used directly to generate correlated increments, avoiding the numerical factorization of 
the correlation matrix. The pairwise correlations $\rho^{ij}$ 
are still computed separately for the penalty term described 
in~\Cref{subsubsec:correlation-penalty}, but this involves only 
the vine recursion and not a full matrix factorization.

\begin{remark}
In the volatility-correlation coordinates used above, the set $B$ is open relative to $\mathbb R^d\times\mathcal E^d$: volatilities belong to the open intervals $(\sigma^i_{\min},\sigma^i_{\max})$, and off-diagonal correlations are strictly between $-1$ and $1$. The boundary values 
$\sigma^i_{\min}$, $\sigma^i_{\max}$, $\pm 1$ are excluded because $\tanh$ maps $\mathbb{R}$ onto the open interval $(-1,1)$. In 
practice, this is not a limitation: boundary configurations can be approximated arbitrarily closely by
$\mathcal T_{\mathrm{UVM}}(z)$ as the relevant latent coordinates tend to
$\pm\infty$. In particular, bang-bang volatility levels, and the low-dimensional 
bang-bang correlation configurations discussed in~\Cref{subsubsec:bsb}, 
can be approached in this limiting sense.
\end{remark}

\subsection{Training}\label{subsec:training}

To avoid ambiguity in the training terminology, we use time step for an index $n$ of the backward time discretization, epoch for one
outer training cycle at a fixed time step $n$, and gradient step for a single parameter update based on one minibatch
within an epoch. At each epoch, a fresh Monte Carlo batch is generated
under the current frozen policy and is then used for the critic and actor
updates.

We now describe the backward training loop implementing the actor-critic scheme. For each time step $n$, the policy parameter $\theta_n$ and the critic parameter $\phi_n$ are trained alternately over a prescribed number of epochs, with one additional critic epoch performed at the end (described below). The parameters at later time steps
\[
(\phi^\star_{n+1},\theta^\star_{n+1}),\dots,(\phi^\star_{N-1},\theta^\star_{N-1})
\]
are kept fixed and serve as already learned continuation values and policies.

This alternate training characterizes the actor-critic method and consists of two coupled optimization stages at each epoch: 
\begin{enumerate}
\item local policy evaluation: update the critic $\mathcal{V}_{\phi_n}$ 
to approximate the value of the current policy $\pi_{\theta^{(\ell)}_n}$,
\item local policy improvement: update the actor $\pi_{\theta_n}$
via the PPO objective~\eqref{eq:PPO-objective}, using advantage estimates
computed from the current batch with the pre-update critic
$\mathcal{V}_{\phi_n^{(\ell)}}$.
\end{enumerate}

\begin{algorithm}[H]
\caption{SPG-UVM algorithm}\label{alg:algo-bwd}
\begin{algorithmic}[1]
\Require Maturity $T$; risk-free rate $r$; spot price $\mathcal{X}_0$; volatility bounds $(\sigma^i_{\min},\sigma^i_{\max})_{1\leq i\leq d}$; correlation bounds $(\rho^{ij}_{\min},\rho^{ij}_{\max})_{1\leq i<j\leq d}$; number of time steps $N$; state-sampling measures $(\mu_n)_{0\leq n\leq N-1}$; number of epochs
\Ensure Trained value functions $(\mathcal{V}_{\phi^{\star}_n})_{0\leq n\leq N-1}$ and trained deterministic Markov policies $(\pi_{\theta^\star_n}^{\mathrm{det}})_{0\leq n\leq N-1}$
\For{$n=N-1$ \textbf{down to} $0$}
  \State Initialize neural networks $\mathcal{V}_{\phi_n}$ and $\pi_{\theta_n}$
  \For{a prescribed number of epochs, indexed by $\ell=0,1,\ldots$}
    \State \textbf{(i) Data collection:} draw $\mathcal{X}^{(\ell)}_{n}\sim\mu_n$, draw $\xi_n^{(\ell)}\sim\mathcal{N}(0,I_d)$, sample $(\sigma_n^{(\ell)},\rho_n^{(\ell)})\sim \pi_{\theta_n^{(\ell)}}(\cdot\mid\mathcal{X}^{(\ell)}_{n})$, and set $a_n^{(\ell)}=\operatorname{diag}(\sigma_n^{(\ell)})L(\rho_n^{(\ell)})$
    \State \textbf{(ii) Critic update:} compute $\phi_n^{(\ell+1)}$ from $\phi_n^{(\ell)}$ by stochastic gradient descent on~\eqref{eq:critic-loss}
    \State \textbf{(iii) Actor update:} compute $\theta_n^{(\ell+1)}$ from $\theta_n^{(\ell)}$ by stochastic gradient ascent on~\eqref{eq:actor-objective}
  \EndFor
  \State \textbf{Final critic epoch:} one additional critic epoch using the deterministic policy induced by $\theta_n^\star$
\EndFor
\end{algorithmic}
\end{algorithm}

\paragraph{Data collection.}
At epoch $\ell$, the actor parameter $\theta_n^{(\ell)}$ is frozen and used to generate training samples. A state $\mathcal{X}^{(\ell)}_n$ is sampled from the distribution $\mu_n$, then volatility-correlation coordinates
\[
(\sigma_n^{(\ell)},\rho_n^{(\ell)})
\sim \pi_{\theta_n^{(\ell)}}(\cdot\mid \mathcal X_n^{(\ell)})
\]
are drawn from the current policy, and finally a Gaussian increment $\xi_n^{(\ell)}\sim\mathcal N(0,I_d)$ is sampled. The control coordinates are generated by first drawing
\[
Z_n^{(\ell)}\sim \mathcal N\big( m_{\theta_n^{(\ell)}}(\mathcal X_n^{(\ell)}), \Lambda^{(\ell)} \big)
\]
and then setting
\[
(\sigma_n^{(\ell)},\rho_n^{(\ell)})
=\mathcal T_{\mathrm{UVM}}(Z_n^{(\ell)}).
\]
The corresponding diffusion matrix used in the transition map is
\[
a_n^{(\ell)}
=\operatorname{diag}(\sigma_n^{(\ell)})L(\rho_n^{(\ell)}).
\]
In what follows, we write $\mathbb{E}^{(\ell)}_n$ for the expectation with respect to the sampling procedure above, namely
$\mathcal X_n^{(\ell)}\sim\mu_n$,
$(\sigma_n^{(\ell)},\rho_n^{(\ell)})\sim\pi_{\theta_n^{(\ell)}}(\cdot\mid\mathcal X_n^{(\ell)})$,
and an independent $\xi_n^{(\ell)}\sim\mathcal N(0,I_d)$. The advantage estimates associated with this batch are computed using the
pre-update critic $\mathcal V_{\phi_n^{(\ell)}}$ and are kept fixed during
the subsequent critic and actor updates.

\paragraph{Critic update.}
The critic is trained to approximate the value of the current policy at step $n$, using the already learned continuation value $\mathcal V_{\phi^\star_{n+1}}$ as a target. The critic loss is therefore
\begin{equation}\label{eq:critic-loss}
\mathcal{L}^{\mathrm C}_n(\phi)\coloneqq
\mathbb E_n^{(\ell)}\left[
\left|
\mathcal V_\phi(\mathcal X_n^{(\ell)})
-
\mathrm e^{-r\frac{T}{N}}
\mathcal V_{\phi^\star_{n+1}}
\big(F(\mathcal X_n^{(\ell)},a_n^{(\ell)},\xi_n^{(\ell)})\big)
\right|^2
\right],
\end{equation}
with the convention $\mathcal V_{\phi^\star_N}\equiv g$ when $n=N-1$.

\paragraph{Actor update.}
The actor update is based on the empirical counterpart of the statewise PPO objective~\eqref{eq:PPO-objective}, averaged over the sampled states. Using the sampled data collected under the frozen policy $\pi_{\theta_n^{(\ell)}}$, the corresponding empirical PPO objective reads
\begin{align*}
\mathcal{L}^{\mathrm{A,PPO}}_n(\theta)\coloneqq
\mathbb E_n^{(\ell)}\Bigg[
\min\Bigg\{&
\frac{p_\theta(\sigma_n^{(\ell)},\rho_n^{(\ell)}\mid \mathcal X_n^{(\ell)})}
{p_{\theta_n^{(\ell)}}(\sigma_n^{(\ell)},\rho_n^{(\ell)}\mid \mathcal X_n^{(\ell)})}
\,
\operatorname{Adv}_n^{(\ell)},\\
&
\operatorname{clip}\Bigg(
\frac{p_\theta(\sigma_n^{(\ell)},\rho_n^{(\ell)}\mid \mathcal X_n^{(\ell)})}
{p_{\theta_n^{(\ell)}}(\sigma_n^{(\ell)},\rho_n^{(\ell)}\mid \mathcal X_n^{(\ell)})},
1-\varepsilon,1+\varepsilon
\Bigg)\operatorname{Adv}_n^{(\ell)}
\Bigg\}
\Bigg],
\end{align*}
where
\[
\operatorname{Adv}_n^{(\ell)}
\coloneqq
\mathrm e^{-r\frac{T}{N}}
\mathcal V_{\phi^\star_{n+1}}
\big(F(\mathcal X_n^{(\ell)},a_n^{(\ell)},\xi_n^{(\ell)})\big)
-
\mathcal V_{\phi_n^{(\ell)}}(\mathcal X_n^{(\ell)}).
\]
Recall that the policy introduced in~\Cref{subsec:policies} relaxes the pairwise correlation bounds of the UVM. The actor objective is thus augmented by a penalty term penalizing violations 
of the pairwise correlation bounds of the UVM. The practical actor objective takes the form
\begin{equation}\label{eq:actor-objective}
\mathcal L_n^{\mathrm A}(\theta)
=
\mathcal L_n^{\mathrm{A,PPO}}(\theta)-\Psi_n(\theta),
\end{equation}
where $\Psi_n$ penalizes violations of the correlation constraints. The precise form of this additional term, together with the corresponding implementation choice, is described in~\Cref{sec:numerics}.

The likelihood ratio appearing in the PPO objective admits an explicit form. At a fixed epoch, the policy used to generate the current batch and the candidate policy are dominated by the same
measure $\kappa$. Since they also share the same latent covariance matrix at
a fixed epoch, their likelihood ratio is the ratio of the corresponding
Gaussian densities:
\begin{align*}
\frac{p_{\theta}(\sigma,\rho\mid x)}
{p_{\theta_n^{(\ell)}}(\sigma,\rho\mid x)}
&=
\frac{
\varphi\big(
\mathcal T_{\mathrm{UVM}}^{-1}(\sigma,\rho);
m_{\theta}(x),\Lambda
\big)}
{
\varphi\big(
\mathcal T_{\mathrm{UVM}}^{-1}(\sigma,\rho);
m_{\theta_n^{(\ell)}}(x),\Lambda
\big)}
\nonumber\\
&=
\exp\Big[
\big(m_{\theta}(x)-m_{\theta_n^{(\ell)}}(x)\big)^\top
\Lambda^{-1}
\big(
Z-\tfrac12(m_{\theta}(x)+m_{\theta_n^{(\ell)}}(x))
\big)
\Big],
\end{align*}
using that $(\sigma,\rho)=\mathcal T_{\mathrm{UVM}}(Z)$, so that
$Z=\mathcal T_{\mathrm{UVM}}^{-1}(\sigma,\rho)$ is directly available
from the latent sample used to generate the control coordinates. Thus,
despite the squashing map, the PPO ratio takes the same explicit Gaussian
form as in the nonsquashed case.

In practice, the critic and actor parameters are updated by stochastic gradient 
descent and stochastic gradient ascent, respectively, using the Adam 
algorithm~\cite{KB14} on empirical minibatch approximations of the critic loss 
and actor objective.

\paragraph{Final critic epoch.}
At the end of the inner loop, we perform one additional critic training epoch,
using the same critic loss as above but with samples generated under the
deterministic policy induced by the final actor parameter $\theta_n^\star$.
We first define its volatility-correlation coordinates by
\[
(\sigma_{\theta_n^\star}^{\mathrm{det}}(x),
 \rho_{\theta_n^\star}^{\mathrm{det}}(x))
\coloneqq
\mathcal{T}_{\mathrm{UVM}}(m_{\theta_n^\star}(x)),
\]
and the corresponding diffusion-matrix representative by
\begin{equation}\label{eq:det-policy}
a_{\theta^\star_n}^{\mathrm{det}}(x)
\coloneqq
\operatorname{diag}\big(
\sigma_{\theta_n^\star}^{\mathrm{det}}(x)
\big)
L\big(
\rho_{\theta_n^\star}^{\mathrm{det}}(x)
\big).
\end{equation}
The corresponding deterministic policy on $B$ is
\[
\pi_{\theta_n^\star}^{\mathrm{det}}(\cdot\mid x)
=
\delta_{(
\sigma_{\theta_n^\star}^{\mathrm{det}}(x),
\rho_{\theta_n^\star}^{\mathrm{det}}(x)
)}.
\]
Therefore, stochastic policies are learned during training with the aim of deploying their deterministic counterpart. This is common practice in stochastic policy-gradient methods, see~\cite{Mon+24} for a theoretical and numerical treatment of it. This yields the final critic parameter $\phi_n^\star$, which approximates the value of the deterministic policy used for inference rather than that of the last stochastic training policy. Note that the pairwise correlation bounds are not enforced structurally by the
deterministic action~\eqref{eq:det-policy}. During training, the penalty term in~\eqref{eq:actor-objective} drives the latent mean of the actor toward configurations satisfying these bounds, so that the deterministic policy is observed to satisfy them in practice when the penalty is appropriately weighted. We refer the reader to Appendix~\ref{app:penalty-sensitivity} for a sensitivity study.

\paragraph{Price estimators.}
The algorithm provides two complementary estimators of the option price at 
time~$0$. The actor price is obtained by Monte Carlo simulation of the 
controlled process under the trained deterministic policy and by averaging 
the associated discounted terminal payoff. When this deterministic policy is 
admissible for the discrete-time UVM, namely when its generated correlations 
satisfy the prescribed pairwise bounds, the actor price is the expected payoff 
under a specific feasible policy and therefore provides a lower bound for the 
optimal value $\mathcal V_0^\star(\mathcal X_0)$.

In our implementation, admissibility with respect to the pairwise correlation 
bounds is encouraged through the penalty term in the actor objective. The 
strength of this penalty is controlled by a parameter $\beta$, see~\Cref{subsubsec:correlation-penalty}. In practice, 
increasing $\beta$ reduces correlation-bound violations, and the sensitivity 
analysis reported in Appendix~\ref{app:penalty-sensitivity} shows that one can choose 
values of $\beta$ for which no material violations of the pairwise correlation bounds are observed
at the numerical precision of the experiments. We therefore report the actor
estimate together with a $95\%$ confidence interval as our primal Monte Carlo
estimate. Its rigorous lower-bound interpretation is the one stated above,
namely it holds whenever the deployed deterministic policy is admissible.

The critic price is obtained by direct evaluation of the learned value function 
at the initial state, $\mathcal V_{\phi^\star_0}(\mathcal X_0)$. Unlike the 
actor price under an admissible deterministic policy, the critic price carries 
no guaranteed bound property and is interpreted as a pointwise approximation 
of $\mathcal V_0^\star(\mathcal X_0)$.

\section{Numerical implementation and experiments}\label{sec:numerics}

In this section, we specify the numerical choices used to implement the actor-critic stochastic policy-gradient scheme introduced in~\Cref{sec:algo}, and then report the resulting pricing experiments. We first describe the main implementation ingredients that are specific to our experiments, namely the treatment of the relaxed correlation constraints, the exploration-exploitation mechanisms, the training design and the computational environment. We then present the benchmark setting and the numerical results in the uncertain- and fixed-correlation regimes.

\subsection{Numerical implementation}\label{subsec:num-impl}

We begin with the implementation choices that complement the general framework of~\Cref{sec:algo}. Some of them are structural, such as the treatment of the relaxed pairwise correlation bounds in the continuous policy class, while others concern the practical stabilization of training and the concrete choice of hyperparameters. We start with the correlation penalty, which is a central ingredient of the continuous-policy implementation.

\subsubsection{Correlation penalty}\label{subsubsec:correlation-penalty}

Recall from~\Cref{subsec:policies} that the parameterized policy is 
supported on the relaxed set $B$ defined in~\eqref{eq:set-B}, in which 
positive definiteness of the correlation matrix is enforced by 
construction, but the pairwise UVM bounds
\[
\rho^{ij}\in[\rho^{ij}_{\min},\rho^{ij}_{\max}],\qquad 1\leq i<j\leq d,
\]
are relaxed. As a consequence, actions sampled from the continuous 
policy may fail to satisfy all correlation constraints of the original 
control problem. To steer the actor toward the admissible region, we 
augment the PPO actor objective with a penalty term.

The parameter $\beta$ plays the role of a classical quadratic penalty 
parameter. Indeed, for a smooth constrained optimization problem
\[
\min_\theta f(\theta)\quad\text{s.t.}\quad g_k(\theta)\leq0,\quad k=1,\dots,m,
\]
standard quadratic penalty results imply, under the usual local regularity 
conditions, that local minimizers of 
\(f(\theta)+\beta\sum_k(g_k(\theta)^+)^2\) approach the constrained local 
solution at rate \(O(\beta^{-1})\) as \(\beta\to\infty\); see 
\cite{Ber76,IS23}. Since the Huber function used below is quadratic near the 
origin, our penalty behaves locally like a quadratic penalty in the 
near-feasible regime. We use this only as an interpretation of the penalty 
weight and do not attempt to verify these regularity conditions for the present 
nonconvex neural-network optimization problem.

We now specify the form of the penalty $\Psi_n$ appearing 
in~\eqref{eq:actor-objective}. Denoting by $\rho_\theta(x)$ the correlation matrix obtained from $\mathcal T_{\mathrm{UVM}}(m_\theta(x))$, we set
\[
\Psi_n(\theta)
\coloneqq
\beta\,\mathbb E_n^{(\ell)}
\left[
\frac{2}{d(d-1)}
\sum_{1\leq i<j\leq d}
\left(
\operatorname{Hub}\left(
\frac{(\rho_\theta^{ij}(\mathcal X_n^{(\ell)})-\rho^{ij}_{\max})^+}
{\rho^{ij}_{\max}-\rho^{ij}_{\min}}
\right)
+
\operatorname{Hub}\left(
\frac{(\rho^{ij}_{\min}-\rho_\theta^{ij}(\mathcal X_n^{(\ell)}))^+}
{\rho^{ij}_{\max}-\rho^{ij}_{\min}}
\right)
\right)
\right],
\]
where $\beta>0$ is the penalty weight and $\operatorname{Hub}$ is the 
Huber function~\cite{Hub64} with transition threshold $\delta>0$:
\[
\operatorname{Hub}(v)\coloneqq
\begin{cases}
\dfrac{v^2}{2\delta}, & \text{if } |v|\leq\delta,\\[6pt]
|v|-\dfrac{\delta}{2}, & \text{if } |v|>\delta.
\end{cases}
\]
This evaluation at the deterministic policy
$\mathcal{T}_{\mathrm{UVM}}\big(m_\theta(\mathcal{X}_n^{(\ell)})\big)$, 
rather than at each noisy sample, ensures that the regularization acts 
directly on the control deployed at inference time. The choice of 
Huber over a pure quadratic is motivated by numerical stability. For 
small violations, the penalty is quadratic and therefore provides a 
smooth gradient signal driving the policy toward feasibility. For 
larger violations, it becomes linear, which prevents excessively large 
gradients and makes training more robust. In our experiments, we set 
$\delta=0.05$ and $\beta=10$. See Appendix~\ref{app:penalty-sensitivity} for a 
sensitivity analysis to $\beta$.

Finally, this penalty mechanism is only needed in dimension $d\geq 3$. 
When $d=1$, there is no correlation variable, and 
$\mathcal{T}_{\mathrm{UVM}}$ reduces to a single scaled $\tanh$ map 
sending $z_\sigma\in\mathbb{R}$ to 
$\sigma^1\in\big(\sigma^1_{\min},\sigma^1_{\max}\big)$. When $d=2$, any 
$\rho^{12}\in(-1,1)$ defines a valid correlation matrix, so the 
positive semidefiniteness constraint is automatic; the $C$-vine map 
reduces to the identity, and the pairwise correlation bound 
$\rho^{12}\in[\rho^{12}_{\min},\rho^{12}_{\max}]$ can be enforced 
directly via a $\tanh$ rescaling, without need for a penalty.

\subsubsection{Exploration and exploitation}\label{subsubsec:exploration-exploitation}

As discussed in~\Cref{subsec:training}, stochastic policies are used 
during training and their deterministic counterpart is deployed at 
inference. Striking a balance between the two requires controlling 
the dispersion of the stochastic policy: large dispersion promotes 
exploration of the action space, small dispersion concentrates the 
policy around the deterministic action used at inference.

This balance is governed by the covariance matrix $\Lambda$ of the 
latent Gaussian. In our experiments, we take at training epoch 
$\ell$
\[
\Lambda^{(\ell)}=\lambda^{(\ell)}I_{\frac{d(d+1)}{2}},
\]
where $\lambda^{(\ell)}>0$ is a scalar temperature parameter. Large 
values of $\lambda^{(\ell)}$ induce significant dispersion of the 
sampled latent variables and therefore promote exploration; as 
$\lambda^{(\ell)}\downarrow 0$, the stochastic policy converges to the 
Dirac measure at the deterministic 
action~\eqref{eq:det-policy}. To pass gradually from 
exploration to exploitation, $\lambda^{(\ell)}$ is annealed during 
training: it is decreased to a small final value, so that 
exploratory behavior dominates early on while near-deterministic 
decisions are recovered at the end. We use a smooth sigmoid-type 
decay, taken from~\cite{PW25} and illustrated 
in~\Cref{fig:annealing}.

\begin{figure}[H]
\centering
\includegraphics[width=0.5\textwidth]{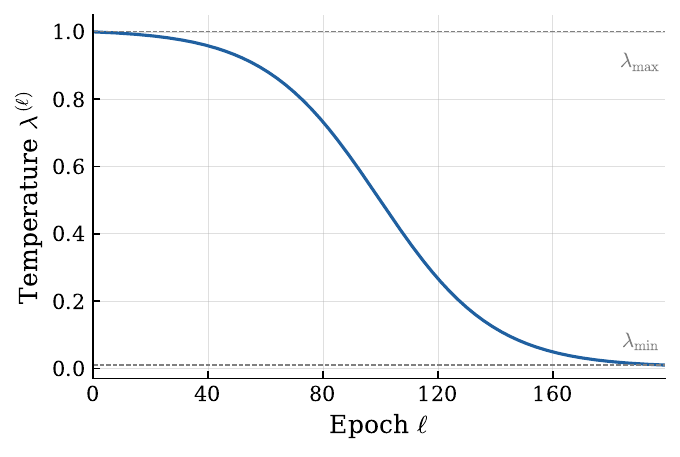}
\caption{Sigmoid annealing schedule of the temperature $\lambda^{(\ell)}$.}
\label{fig:annealing}
\end{figure}

\subsubsection{Training design and neural-network architecture}\label{subsubsec:training-design}

We now describe the concrete training design used in our numerical experiments. The actor and the critic are both represented by feedforward neural networks with one hidden layer, using the nonlinear $\operatorname{LeakyReLU}$ activation function
\[
\operatorname{LeakyReLU}(x)\coloneqq
\max\{0,x\} + 0.01\min\{0, x\}
\]
on the hidden layer and the identity activation on the output layer. Both networks take as input the current state $x\in\mathbb R^d$. The critic outputs a scalar, while the actor outputs a vector in $\mathbb R^{\frac{d(d+1)}{2}}$ when both volatilities and correlations are uncertain, or in $\mathbb R^d$ when correlations are fixed and only the volatility components are controlled. In all the experiments reported below, both actor and critic use $32$ neurons in their hidden layer.

The time discretization parameter $N$ in~\eqref{eq:time-discretization} depends on the considered example and is specified for each payoff below. During training, we use $2^{16}$ Monte Carlo samples per epoch, while actor-based price estimates are computed with $2^{19}$ Monte Carlo samples and are reported with a $95\%$ confidence interval. During training, the Gaussian increments are sampled in antithetic pairs, 
that is, each draw $\xi$ is paired with $-\xi$. This reduces the variance 
of the Monte Carlo estimators without changing the nominal batch size.

The actor and critic parameters are optimized with the Adam algorithm~\cite{KB14},
using minibatches of size $2^{11}$. Since each epoch uses $2^{16}$ Monte Carlo
samples, one pass through the epoch batch corresponds to
$2^{16}/2^{11}=32$ minibatch gradient steps for each network. We deliberately use large minibatches: this improves GPU efficiency and, in our experiments, leads to more stable optimization by reducing batch noise.

At each time step, training runs for a fixed number of epochs, chosen to align with the sigmoid annealing schedule of the exploration parameter. This design reflects the two-phase structure imposed by the annealing schedule of \Cref{subsubsec:exploration-exploitation}: a fixed epoch budget ensures that both the exploration and exploitation phases complete fully, which would not be guaranteed by an early stopping criterion based on the critic loss, since the critic loss may plateau during the exploratory phase before the policy has begun to exploit. Actor and critic learning rates are set to $5\times10^{-3}$. The temperature of the policy is annealed from $1$ to $0.01$. The PPO clipping parameter is fixed to $\varepsilon=0.2$, and the advantages are normalized to have zero mean and unit variance within each epoch, following standard PPO practice.

It remains to specify the state-sampling distribution $\mu_n$ introduced in~\Cref{sec:algo}. We construct the state-sampling distribution $\mu_n$ by sampling the volatility components independently and uniformly in their admissible intervals, while setting correlations to $0$. More precisely, the volatility components are first sampled independently
and uniformly over their admissible intervals, and the components of
$\mathcal X_n^{(\ell)}$ are then generated independently as
\[
\mathcal{X}_0^i\exp\Big(\big(r-\tfrac{(\sigma^i)^2}{2}\big)t_n+\sigma^i\sqrt{t_n}\,Y^i\Big),
\qquad
Y\sim\mathcal N(0,I_d),\quad
\sigma^i\sim\mathcal U\big([\sigma^i_{\min},\sigma^i_{\max}]\big).
\]
This sampling distribution is not intended to reproduce the exact law of the 
optimally controlled state. Its role is to provide a broad, inexpensive coverage 
of the state region visited by plausible volatility scenarios. Ignoring correlations in the state-sampling distribution is a deliberate simplification. In dimensions $d\geq 3$, it avoids repeated matrix factorizations when generating correlated samples and therefore reduces computational cost. In our experiments, this simplification does not lead to any noticeable degradation in training quality.

Finally, we use transfer learning across time steps. More precisely, after training the actor and critic at time step $n+1$, we initialize the networks at time step $n$ with the corresponding trained weights. This exploits the similarity of neighboring backward subproblems and provides
an effective warm start for the optimization at time step $n$. In practice, transfer learning significantly accelerates the overall backward procedure: after a longer training at the first backward step $n=N-1$ with 
$500$ epochs, subsequent time steps require far fewer epochs. Accordingly, the epoch budget is dropped to $10$ for all steps $n\leq N-2$, and the learning rates are halved. Similar observations have been reported in related backward-learning approaches, see for instance~\cite{BHLP22}.

\begin{remark}
We fix the epoch budget as above for all tested payoffs and dimensions, without adapting it to the specific problem. This choice provides a good balance between 
accuracy and computational cost across the panel of options considered. In preliminary tests, this budget captured the main accuracy gains while keeping the total backward training time moderate.
\end{remark}

\subsubsection{Implementation details}\label{subsubsec:implementation-details}

The training and pricing phases are implemented in PyTorch and executed on a single GPU. All state variables, Gaussian increments and neural-network parameters are stored as tensors on a single CUDA device, with parallelism exploited along the Monte Carlo and batch dimensions: at each time step, the controlled Euler scheme advances all simulated paths in a single vectorized operation, and the actor-critic networks are evaluated on whole minibatches of states. Time remains sequential at the Python level, while the corresponding low-level scheduling, kernel execution and matrix operations are handled by PyTorch on the GPU.

Numerical experiments were run on a Dell Precision 7680 mobile workstation equipped with an Intel Core i9-13950HX CPU (24 cores, 32 threads, 2.20\,GHz), 64\,GB RAM and an NVIDIA GeForce RTX~4090 GPU (16\,GB VRAM), running Microsoft Windows~11 Pro. All code was implemented in Python~3.12.8 using PyTorch~2.6.0 with CUDA~12.6 and cuDNN~9.5.1.

The implementations reported in~\cite{GMZ26}, which we use as our main benchmark, are CPU-based and their runtimes mostly range from a few dozen minutes to several hours per experiment. By contrast, our training loop runs on a single local GPU and yields runtimes of the order of minutes in the experiments below. For this reason, we refrain from making direct quantitative runtime comparisons with~\cite{GMZ26}. We nevertheless report our own runtimes, since they are obtained with a 
straightforward PyTorch implementation, without custom CUDA kernels or hardware-specific low-level tuning. The reported runtimes should therefore be interpreted as practical computational costs rather than lower bounds on what could be achieved with a highly optimized implementation.

\subsection{Experimental setting}\label{subsec:exp-setting}

We now describe the experimental setting used to assess the proposed scheme. We first explain the benchmark choice, then specify the classes of options and model parameters considered, and finally describe the content of the summary tables reported below.

\paragraph{Benchmark choice.}

The literature on high-dimensional option pricing in the UVM is sparse. To the best of our knowledge, the main works reporting extensive high-dimensional pricing experiments in this framework are~\cite{GH11} and~\cite{GMZ26}. Since~\cite{GMZ26} reports a broader set of high-dimensional UVM experiments and includes machine-learning-based methods, we use its results as our main benchmark.

More precisely,~\cite{GMZ26} introduce two numerical schemes. The first one, denoted by GTU, combines Gaussian process regression with a multidimensional tree. It is generally the most accurate of the two, but also the most computationally expensive in high dimension because of its inherent tree structure. The second one, denoted by NNU, parameterizes the optimal control globally by a time- and state-dependent neural network trained forward in time. It scales better with the dimension and, like our method, provides approximate controls together with a price estimate. These two methods therefore constitute natural and complementary benchmarks for the present work.

When comparing our prices with those reported in~\cite{GMZ26}, we proceed as follows. For a common choice of the time discretization parameter $N$, we report the GTU price associated with the largest number of Gaussian-process regression points, and the NNU price associated with the largest number of training epochs. These correspond to the highest-resolution configurations reported in~\cite{GMZ26}.

\paragraph{Tested options and model parameters.}

The considered options are the geo-outperformer, geo-call spread, outperformer spread and best-of butterfly options. In the uncertain-correlation setting, we test the geo-outperformer up to dimension $d=5$ and include two-dimensional outperformer spread and best-of butterfly examples. In the fixed-correlation setting, we test the geo-call spread up to dimension $d=80$. We also include a one-dimensional path-dependent example, namely the call Sharpe option.

Unless explicitly stated otherwise, the model parameters are those reported in~\Cref{tab:all_dim_params}. The best-of butterfly is tested under the modified parameters of~\Cref{tab:new_params}.

\begin{table}[H]
  \centering
  \begin{tabular}{ccccccc}
    \toprule
    $\mathcal{X}_0^i$ & $\sigma_{\min}^i$ & $\sigma_{\max}^i$ & $\rho^{ij}_{\min}$ & $\rho^{ij}_{\max}$ & $T$ & $r$ \\
    \midrule
    $100$ & $0.1$ & $0.2$ & $-0.5$ & $0.5$ & $1$ & $0$ \\
    \bottomrule
  \end{tabular}
  \caption{Default model parameters used in the numerical experiments.}
  \label{tab:all_dim_params}
\end{table}

\begin{table}[H]
  \centering
  \begin{tabular}{ccccccc}
    \toprule
    $\mathcal{X}_0^i$ & $\sigma_{\min}^i$ & $\sigma_{\max}^i$ & $\rho^{ij}_{\min}$ & $\rho^{ij}_{\max}$ & $T$ & $r$ \\
    \midrule
    $100$ & $0.3$ & $0.5$ & $0.3$ & $0.5$ & $0.25$ & $0.05$ \\
    \bottomrule
  \end{tabular}
  \caption{Model parameters used for the best-of butterfly.}
  \label{tab:new_params}
\end{table}

\paragraph{Tables description.}

The results produced by our scheme are summarized in~\Cref{tab:uc-summary} and~\Cref{tab:cc-summary}. The columns Actor price and Critic price report the two price estimators introduced in~\Cref{subsec:training}. The column Runtime includes the full backward training loop over all time steps together with the final Monte Carlo simulation used to compute the actor-based estimate. Finally, Reference price denotes a reliable reference value, computed by a payoff-dependent benchmark method specified in each case below.

For comparison purposes, \Cref{tab:uc-bench} and~\Cref{tab:cc-bench} report the benchmark values from~\cite{GMZ26}, together with the corresponding reference prices.

\subsection{Uncertain correlation}\label{subsec:uc}

When correlations are uncertain, the control dimension is $\frac{d(d+1)}{2}$ so the problem rapidly becomes challenging as the number of assets increases. As in~\cite{GMZ26}, we therefore restrict to dimensions up to $d=5$, which already corresponds to a $15$-dimensional control problem. We report in~\Cref{tab:uc-summary} the results obtained with our scheme, and in~\Cref{tab:uc-bench} the GTU and NNU benchmark values from~\cite{GMZ26}. For all uncertain-correlation experiments, the number of time steps is fixed to $N=128$.

\begin{table}[H]
\centering
\caption{Uncertain-correlation tests: estimated prices, runtimes, and reference prices.}
\label{tab:uc-summary}
\begin{adjustbox}{max width=\textwidth}
\begin{tabular}{@{} l c c c c c @{}}
\toprule
\textbf{Option} & $d$ & Actor price & Critic price & Runtime (s) & Reference price \\
\midrule
\multirow{4}{*}{\textbf{Geo-outperformer}}
& 2 & $13.75 \pm 0.06$ & 13.77 & 136 & 13.75 \\
& 3 & $12.92 \pm 0.05$ & 13.01 & 214 & 12.96 \\
& 4 & $12.62 \pm 0.05$ & 12.75 & 291 & 12.73 \\
& 5 & $12.50 \pm 0.05$ & 12.61 & 365 & 12.64 \\
\midrule
\textbf{Outperformer spread}
& 2 & $12.79 \pm 0.02$ & 12.85 & 141 & 12.83 \\
\midrule
\textbf{Best-of butterfly}
& 2 & $6.62 \pm 0.02$ & 6.62 & 163 & 6.70 \\
\bottomrule
\end{tabular}
\end{adjustbox}
\end{table}

\begin{table}[H]
\centering
\caption{Uncertain-correlation tests: \cite{GMZ26} benchmarks.}
\label{tab:uc-bench}
\begin{adjustbox}{max width=\textwidth}
\begin{tabular}{@{} l c c c c @{}}
\toprule
\textbf{Option} & $d$ & GTU & NNU & Reference price \\
\midrule
\multirow{4}{*}{\textbf{Geo-outperformer}}
& 2 & 13.76 & $13.70 \pm 0.13$ & 13.75 \\
& 3 & 12.93 & $12.84 \pm 0.12$ & 12.96 \\
& 4 & 12.69 & $12.62 \pm 0.11$ & 12.73 \\
& 5 & 12.54 & $12.45 \pm 0.11$ & 12.64 \\
\midrule
\textbf{Outperformer spread}
& 2 & 12.77 & $12.80 \pm 0.05$ & 12.83 \\
\bottomrule
\end{tabular}
\end{adjustbox}
\end{table}

We begin with the geo-outperformer option, whose payoff is
\[
\left(\sqrt[d-1]{\prod_{i=2}^d X_T^i}-X_T^1\right)^+.
\]
For this option, the reference price is the high-accuracy reference reported in~\cite{GMZ26}, obtained by fixing the correlations at their known optimal values and then maximizing over the volatilities with a refined GTU run. Our algorithm remains stable across all tested dimensions $d=2,\dots,5$. As the dimension increases, the actor-based estimate deteriorates slightly, with a relative error of $1.11\%$ at $d=5$, a behavior also observed for NNU, whereas GTU remains below the $1\%$ level. The critic-based estimate is also close to the reference value throughout, and becomes particularly accurate in the higher-dimensional cases $d=4,5$. The runtimes also remain moderate, with about six minutes for the $d=5$ case.

We next consider the outperformer spread option
\[
\left(X_T^2-0.9X_T^1\right)^+-\left(X_T^2-1.1X_T^1\right)^+,
\]
for which the reference price is computed by the finite-difference method of~\cite{MF17}. Here again, both actor and critic estimates are accurate, with an actor-based relative error of $0.31\%$. The critic-based estimate is slightly more accurate, and the runtime remains short.

Finally, we consider the best-of butterfly
\[
\left(X_T^{\max}-K_1\right)^+
-2\left(X_T^{\max}-\frac{K_1+K_2}{2}\right)^+
+\left(X_T^{\max}-K_2\right)^+,
\qquad
X_T^{\max}\coloneqq \max\{X_T^1,X_T^2\},
\]
with $K_1=85$ and $K_2=115$. This test case is not reported in~\cite{GMZ26}, so we use the modified parameters of~\Cref{tab:new_params}, and compute the reference price with the finite-difference scheme of~\cite{MF17}. This example is more challenging because the payoff is nonconvex and the larger volatility interval increases the range of admissible diffusion coefficients. Nevertheless, the algorithm remains accurate: the actor-based relative error is $1.19\%$, and the critic-based relative error is $1.19\%$.

\begin{figure}[H]
\centering
\includegraphics[width=1.0\textwidth]{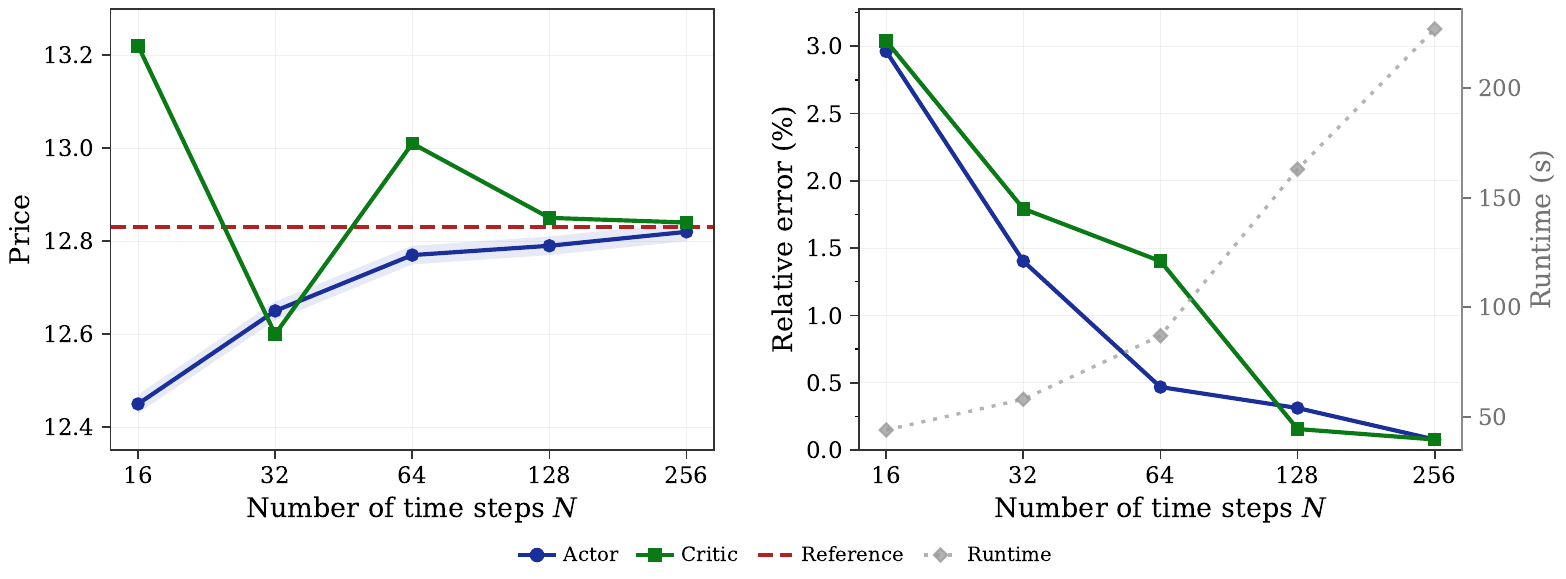}
\caption{Convergence with respect to the number of time steps $N$ for the 
outperformer spread option. The left panel reports actor and critic prices, 
together with the reference price; the right panel reports absolute relative 
errors.}
\label{fig:outperf-CV}
\end{figure}

\begin{figure}[H]
\centering
\includegraphics[width=1.0\textwidth]{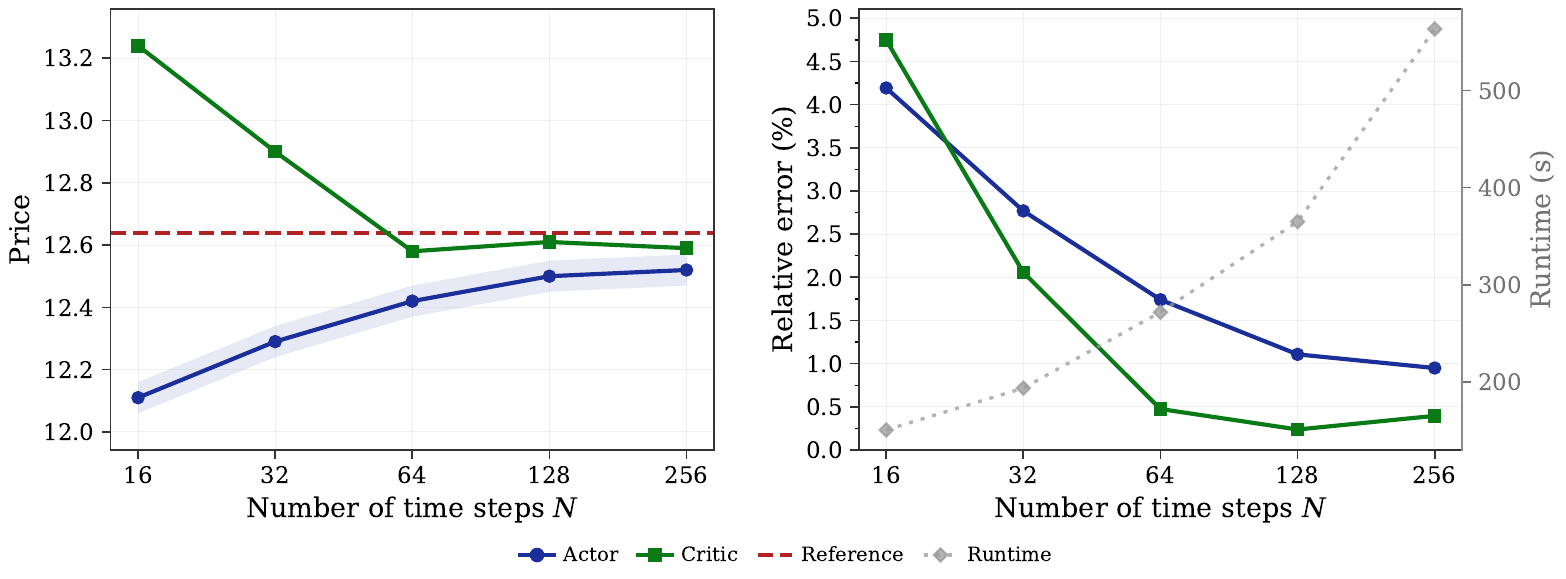}
\caption{Convergence with respect to the number of time steps $N$ for the 
geo-outperformer option with $d=5$. The left panel reports actor and critic 
prices, together with the reference price; the right panel reports absolute 
relative errors.}
\label{fig:geo-outperf-CV}
\end{figure}

We also report in~\Cref{fig:outperf-CV} and \Cref{fig:geo-outperf-CV} the evolution of critic and actor prices, as well as their associated relative errors, for the outperformer spread and geo-outperformer ($d=5$) options. Relative errors are defined as the absolute value of
\[
\frac{\text{estimated price} - \text{reference price}}
{\text{reference price}}.
\]
The left panels illustrate a structural difference between the two price estimators. Empirically, the actor price approaches the reference from below as $N$ increases. 
This behavior is consistent with the lower-bound property of feasible actor 
policies and with the convergence of the discrete-time value 
$\mathcal V_0^\star$ to the continuous-time value established in 
\Cref{lem:convergence-rate}. The critic price, by contrast, is a pointwise evaluation of the learned value function and carries no guaranteed bias direction: it may lie above or below the reference. The right panels confirm that the absolute relative errors of both estimators decrease consistently with $N$, which is the relevant convergence diagnostic. We note that the 
sign of the critic bias may vary across independent runs, while its magnitude decreases steadily with $N$; additional convergence plots for other payoffs, confirming the same behavior, are reported 
in Appendix~\ref{app:convergence-plots}.

\subsection{Fixed correlation}\label{subsec:cc}

The fixed-correlation setting is simpler than the uncertain-correlation one, because the control acts only through the $d$ volatility components, so we consider a smaller variety of payoffs but explore larger dimensions. In addition, we include a one-dimensional path-dependent example, namely the call Sharpe option. We report in~\Cref{tab:cc-summary} the results obtained with our scheme, and in~\Cref{tab:cc-bench} the GTU and NNU benchmark values from~\cite{GMZ26}. For the geo-call spread option, we use $N=64$ for $d=2,5,10$ and $N=32$ in higher dimensions. For the call Sharpe option, we use $N=192$.

\begin{table}[H]
\centering
\caption{Fixed-correlation tests: estimated prices, runtimes, and reference prices. Confidence half-widths reported as $0.00$ are below \(0.005\) after rounding.}
\label{tab:cc-summary}
\begin{adjustbox}{max width=\textwidth}
\begin{tabular}{@{} l c c c c c @{}}
\toprule
\textbf{Option} & $d$ & Actor price & Critic price & Runtime (s) & Reference price \\
\midrule
\multirow{6}{*}{\textbf{\shortstack{Geo-call spread\\ ($\rho=0$)}}}
& 2  & $10.51 \pm 0.02$ & 10.49 & 71 & 10.50 \\
& 5  & $9.68 \pm 0.01$  & 9.70 & 83 & 9.70 \\
& 10 & $9.55 \pm 0.01$  & 9.55 & 41 & 9.55 \\
& 20 & $9.53\pm 0.00$  & 9.53 & 46 & 9.53 \\
& 40 & $9.51\pm 0.00$  & 9.51 & 53 & 9.51 \\
& 80 & $9.51\pm 0.00$  & 9.52 & 64 & 9.51 \\
\midrule
\textbf{Call Sharpe}
& 1 & $57.22 \pm 0.19$ & 57.39 & 320 & 58.40 \\
\bottomrule
\end{tabular}
\end{adjustbox}
\end{table}

\begin{table}[H]
\centering
\caption{Fixed-correlation tests: \cite{GMZ26} benchmarks.}
\label{tab:cc-bench}
\begin{adjustbox}{max width=\textwidth}
\begin{tabular}{@{} l c c c c @{}}
\toprule
\textbf{Option} & $d$ & GTU & NNU & Reference price \\
\midrule
\multirow{6}{*}{\textbf{\shortstack{Geo-call spread\\ ($\rho=0$)}}}
& 2  & 10.49 & $10.46 \pm 0.05$ & 10.50 \\
& 5  & 9.74  & $9.66 \pm 0.03$ & 9.70 \\
& 10 & 9.52  & $9.53 \pm 0.02$ & 9.55 \\
& 20 & 9.56  & $9.53 \pm 0.01$ & 9.53 \\
& 40 & 9.54  & $9.51 \pm 0.01$ & 9.51 \\
& 80 & 9.27  & $9.50 \pm 0.01$ & 9.51 \\
\midrule
\textbf{Call Sharpe}
& 1 & 57.90 & $56.32 \pm 0.43$ & 58.40 \\
\bottomrule
\end{tabular}
\end{adjustbox}
\end{table}

We first consider the geo-call spread option, with payoff
\[
\left(\sqrt[d]{\prod_{i=1}^d X_T^i}-K_1\right)^+
-
\left(\sqrt[d]{\prod_{i=1}^d X_T^i}-K_2\right)^+,
\]
where $K_1=90$, $K_2=110$, and the correlations are fixed to $0$. The reference price is taken from~\cite{GMZ26}, where the problem is reduced to one dimension and priced by Monte Carlo. This example is numerically easier than the uncertain-correlation tests: both actor and critic estimates are extremely accurate, with relative errors below $0.3\%$ in all reported dimensions, and runtimes remain of the order of one minute. GTU and NNU achieve comparable performance, except at $d=80$, where GTU loses noticeable accuracy.

We next consider the call Sharpe option, whose payoff at maturity is
\[
\frac{(X_T-K)^+}{\sqrt{V_T}},
\]
where $K=100$ and
\[
V_T=\frac{1}{T}\sum_{n=1}^{N_m}\left(\ln\frac{X_{\tau_n}}{X_{\tau_{n-1}}}\right)^2
\]
is the realized variance computed from monthly returns, that is, $N_m=12T$, $\tau_0=0$, and $\tau_n=n/12$ for $n=1,\dots,N_m$. At time $t$, the option value depends not only on $X_t$, but also on the path-dependent quantities
\[
A_t^1\coloneqq \sum_{n:\,\tau_n\le t}\left(\ln\frac{X_{\tau_n}}{X_{\tau_{n-1}}}\right)^2,
\qquad
A_t^2\coloneqq X_{\sup\{\tau_n:\,\tau_n\le t\}}.
\]
A Markov representation is therefore obtained by augmenting the state with $(A_t^1,A_t^2)$. To ensure that the monitoring dates are contained in the time grid, we choose $N=192$, which is a multiple of $N_m$. The reference price is computed in~\cite{GH11} by a PDE method. This is a more challenging example because of the path-dependent payoff and the augmented state representation. The degradation in accuracy is visible for all methods: in particular, NNU exhibits a relative error of about $3.6\%$. Our method also deteriorates, but remains reasonably accurate, with actor and critic relative errors of $2.02\%$ and $1.73\%$ respectively.

\section{Superhedging strategy and dual diagnostics}\label{sec:dual}

The actor-critic scheme of~\Cref{sec:algo} approximates the discrete-time
robust price $\mathcal V^\star_0(\mathcal X_0)$ through the actor and critic
estimators. In addition, the spatial gradient of the critic provides an
approximation of the associated Markov superhedging strategy. The quality of
this learned gradient is therefore of direct financial and numerical interest.
In this section, we assess it through the dual formulation of stochastic control
developed in~\cite{HLR16}. The critic gradient is used as the martingale
integrand in the corresponding discretized dual construction, so that the
resulting dual estimate provides a diagnostic of the learned gradient.The underlying exact dual functional yields an upper bound for the robust
price. Since the inner pathwise supremum is approximated numerically below,
we refer to the implemented quantity as a dual estimate. We first
recall the duality result, then introduce the Sobolev Training procedure used to
improve the critic gradient, and finally describe the numerical resolution of
the inner pathwise maximization. We conclude with numerical experiments
comparing the gradients produced by classical and Sobolev critic training.

\subsection{Duality result}\label{subsec:dual}

In~\cite{HLR16}, the authors derive a dual formulation \`a
la~\cite{DB92,Rog07} for stochastic control problems on a canonical Wiener
space. Their Markovian duality theorem is proved under bounded-coefficient
assumptions and therefore does not apply directly to
\eqref{eq:UVM-multi-dynamics}. Nevertheless, the same dual construction is
used in~\cite{HLR16} for the UVM and provides the basis for the numerical
diagnostic considered here. In the setting of~\cite[Theorem~2.6]{HLR16},
the dual representation takes the form
\begin{equation}\label{eq:hlr-dual}
V(0,\mathcal{X}_0)\;=\;
\lim_{h\to 0}\,U_h(0,\mathcal{X}_0),
\end{equation}
with
\begin{equation}\label{eq:U_h}
    U_h(0,\mathcal{X}_0)\coloneqq\inf_{\varphi\in\mathcal{U}}\,
\mathbb{E}\Big[\max_{a\in\mathcal{D}_h}\big\{\,\mathrm{e}^{-rT}\,g(X_T^a)
- \int_0^T \mathrm{e}^{-rt}\,
\varphi_t(X^a)^\top
\operatorname{diag}(X_t^a)\,a_t\,dW_t\big\}\Big].
\end{equation}
In~\eqref{eq:U_h}, the inner maximum is taken over the set of piecewise constant functions
\[
\mathcal{D}_h\coloneqq\big\{f:[0,T]\to\mathbb{R}^{d\times d}:f\text{ is constant on } [t_k^h, t_{k+1}^h) \text{ for all }k\text{ and takes values in } A_h\big\},
\]
where $(t_k^h)_{k}$ is a finite partition of $[0,T]$ such that $t^h_{k+1}-t_k^h\leq h$ and $A_h$ is a finite $h$-net of $A$, i.e., for any $a\in A$, there exists $b\in A_h$ such that $|a-b|\leq h$. The outer infimum is taken over the set of processes
\[
\mathcal{U}\coloneqq\big\{\varphi:[0,T]\times\Omega\to\mathbb{R}^d:\varphi\text{ is bounded and }(\mathcal{F}_s^W)_{s\in[0,T]}\text{-adapted}\big\}.
\]

The duality theorem gives the limiting identity. For our numerical purposes,
the relevant point is that the dual construction relies on a piecewise-constant
policy approximation. Following the sandwich argument of
\cite[Remark~2.9]{HLR16}, the time-discretization effect can be related to the
corresponding policy approximation error. The sharper estimate of
\cite[Theorem~2.1]{JPR19}, already used in the proof of
\Cref{lem:convergence-rate}, gives an order $h^{1/4}$ for the
piecewise-constant policy approximation under~\Hyp{H2prime} and~\Hyp{H4}.
Hence, with $h=T/N$, the quantity $(T/N)^{1/4}$ provides a natural reference
scale for the residual time-discretization gap observed in the dual experiments.
We use this rate only as a benchmark for the time-discretization component,
since the implemented dual estimate also involves a numerical approximation of
the inner pathwise maximization.

For numerical purposes and following~\cite{HLR16}, we evaluate the dual objective on the time grid~\eqref{eq:time-discretization} and use a prescribed martingale integrand rather than optimizing over 
$\varphi$. This yields a computable discretized dual estimate. Given 
any sequence of functions $(\varphi_n)_{0\leq n\leq N-1}$ with 
$\varphi_n:(0,+\infty)^d\to\mathbb{R}^d$, we obtain the following discretized dual functional:
\begin{equation}\label{eq:hlr-discrete-dual}
\mathbb E\left[
\sup_{a\in A^N}
\Big\{
\mathrm e^{-rT}g(\mathcal X_N^a)
-
\sum_{n=0}^{N-1}
\varphi_n(\mathcal X_n^a)^\top
\big(
\mathrm e^{-rt_{n+1}}\mathcal X_{n+1}^a
-
\mathrm e^{-rt_n}\mathcal X_n^a
\big)
\Big\}
\right].
\end{equation}
The discrete martingale correction in~\eqref{eq:hlr-discrete-dual} is the
natural Riemann-sum analogue of the stochastic integral in~\eqref{eq:U_h},
expressed in terms of discounted price increments. By~\cite[Remark~2.8]{HLR16}, the natural 
choice for $\varphi$ is $\varphi^\star_t(\omega)\coloneqq\nabla_x 
V(t,\omega_t)$, which suggests using 
$\varphi_n=\nabla_x\mathcal{V}_{\phi_n^\star}$, the gradient of the 
learned critic obtained by automatic differentiation.

\subsection{Sobolev Training of the critic}\label{subsec:sobolev}

The critic is trained by regressing $\mathcal{V}_{\phi_n}$ against value targets: the resulting loss~\eqref{eq:critic-loss} constrains $\mathcal{V}_{\phi_n}$ but provides no guarantee on the quality of its spatial derivatives. In practice, there is no reason for the critic to approximate well the gradient of the value function, which can in turn severely impact the dual estimate since the martingale correction in~\eqref{eq:hlr-discrete-dual} accumulates $N$ gradient evaluations along the path. To address this, we augment the critic training with an explicit gradient objective, following the Sobolev Training paradigm of~\cite{Cza+17}, which is inspired by the universal approximation results for neural networks in Sobolev spaces of~\cite{HSW90}.

In our setting, both value targets and sampled gradient targets can be constructed at each backward step from the dynamic programming recursion, making Sobolev Training directly applicable. Precisely, recall that the target used in the critic training regression loss~\eqref{eq:critic-loss} is 
\[
\mathrm e^{-r\frac{T}{N}}
\mathcal V_{\phi^\star_{n+1}}
\big(F(\mathcal X_n^{(\ell)},a_n^{(\ell)},\xi_n^{(\ell)})\big),
\]
with the convention $\mathcal V_{\phi^\star_N}\equiv g$ when $n=N-1$. Using that the log-Euler scheme~\eqref{eq:log-euler} is differentiable with respect to the state variable $x$, and assuming the critic $\mathcal{V}_{\phi^\star_{n+1}}$ to be differentiable (which is the case if the activation function used in its neural network is differentiable), first-order targets are thus computable in the spirit of~\cite{HS20} via the chain rule
\[
\mathrm e^{-r\frac{T}{N}}\,
(\partial_x F)^\top\big(\mathcal X_n^{(\ell)},a_n^{(\ell)},\xi_n^{(\ell)}\big)\;\nabla_x\mathcal V_{\phi^\star_{n+1}}
\big(F(\mathcal{X}_n^{(\ell)},a_n^{(\ell)},\xi_n^{(\ell)})\big).
\]
In the above expression, the action $a_n^{(\ell)} = a_n^{(\ell)}\big(\mathcal{X}^{(\ell)}_n\big)$ is treated as fixed (i.e., independent of $x$) when 
differentiating the transition map $F$ with respect to the state $x$. Thus 
$\partial_xF$ denotes the derivative of $F$ with respect to its first argument 
only, at fixed action and noise. This choice is consistent with the envelope 
principle underlying dynamic programming: for the exact Bellman operator, the 
state derivative of the maximized value is obtained by differentiating the 
one-step objective with respect to the state at an optimal action, without 
differentiating the optimal action itself. In the numerical scheme the actor is 
only approximate, so the resulting gradient target is biased. We nevertheless 
use this partial-derivative target because it is more stable than 
backpropagating through the actor and the $C$-vine transformation, and because the quality of the resulting gradient is subsequently assessed through the dual estimate.

The critic loss is augmented as follows:
\begin{equation}\label{eq:sobolev-critic-loss}
\begin{aligned}
&\mathcal{L}^{\mathrm{C}}_n(\phi)\\
&\quad+
\eta_{\mathrm{dV}}\,
\mathbb{E}^{(\ell)}_n\Bigg[
\Bigg|
\nabla_x\mathcal{V}_\phi(\mathcal{X}_n^{(\ell)})
-
\mathrm e^{-r\frac{T}{N}}\,
(\partial_x F)^\top
\big(\mathcal X_n^{(\ell)},a_n^{(\ell)},\xi_n^{(\ell)}\big)
\nabla_x\mathcal V_{\phi^\star_{n+1}}
\big(F(\mathcal{X}_n^{(\ell)},a_n^{(\ell)},\xi_n^{(\ell)})\big)
\Bigg|^2
\Bigg],
\end{aligned}
\end{equation}
where $\eta_{\mathrm{dV}}>0$ balances the value and gradient components of the critic loss. In practice, each term in~\eqref{eq:sobolev-critic-loss} is normalized by the empirical mean square of its own target, computed per backward step and treated as a constant (no gradient flows through it). This rescaling puts both terms on a comparable scale, so that $\eta_{\mathrm{dV}}$ plays a single role (assigning a relative weight to the derivative term) rather than the double role of both compensating for magnitude differences between the value and gradient targets and weighting the gradient loss. Derivatives are efficiently computed via automatic differentiation. At $n=N-1$, the target involves the payoff $g$. For the piecewise smooth payoffs and neural networks used in our experiments, we use the a.e. derivatives returned by automatic differentiation.

The learned gradient \(\nabla_x\mathcal V_{\phi_n^\star}\) therefore serves a double purpose: it provides an approximate superhedging strategy and the martingale integrand used in the discretized dual construction. The resulting dual estimate provides a numerical diagnostic of the quality of this gradient.

\subsection{Numerical experiments}

In all dual experiments below, the Sobolev critic is trained with 
$\eta_{\mathrm{dV}}=100$ in~\eqref{eq:sobolev-critic-loss}. We also increase the number 
of hidden neurons in the critic network from $32$ to $256$, while keeping all the other parameters as in~\Cref{sec:numerics}. This increased capacity helps the 
critic fit both value and gradient targets.

\subsubsection{Inner maximum computation}\label{subsubsec:inner-max}

The discretized dual objective~\eqref{eq:hlr-discrete-dual} requires,
for each Monte Carlo realization, the computation of a pathwise supremum over the whole discretized control sequence. More precisely, for a fixed
Monte Carlo path $m$, we first freeze the Brownian increments
$(\xi_n^{(m)})_{0\leq n\leq N-1}$. The state trajectory then becomes a
deterministic function of the control sequence
\[
u=(u_0,\ldots,u_{N-1}),
\qquad
u_n\in(-1,1)^{d(d+1)/2},
\]
through the post-squashing volatility rescaling and $C$-vine construction
introduced in~\Cref{subsec:policies}, together with the forward scheme. Hence, for each path, the inner problem is a deterministic
finite-dimensional optimization problem over
$N\,d(d+1)/2$ scalar variables, solved independently of the other Monte
Carlo paths. In particular, this is not an optimization of an expectation
and not a feedback-policy optimization: the entire noise realization is
fixed before the pathwise maximization is performed.

Rather than optimizing directly over the physical action set $A$, which
includes positive-semidefiniteness constraints on the correlation matrix,
we work in the squashed latent coordinates
$u=\tanh(z)\in(-1,1)^{d(d+1)/2}$ and map them to physical controls using
the volatility rescaling and $C$-vine construction introduced in
\Cref{subsec:policies}. Given a mesh size $\delta>0$, each latent coordinate
is restricted to a uniform one-dimensional mesh
$\mathcal G_\delta\subset(-1,1)$ with spacing at most $\delta$. The continuous pathwise supremum in~\eqref{eq:hlr-discrete-dual} is thus
approximated by restricting the latent control sequence to
\[
u\in
\mathcal G_\delta^{\,N d(d+1)/2}.
\]
The volatility bounds are enforced by construction. For $d\geq3$, the same normalized Huber penalty as in
\Cref{subsubsec:correlation-penalty} is used to penalize pairwise-correlation
violations along the path.

In~\cite{HLR16}, the corresponding pathwise maximization is treated as a
derivative-free deterministic optimization problem; for their UVM
experiments, the authors report using a direct-search polytope algorithm.
Here we use instead a GPU-vectorized cyclic coordinate-ascent procedure,
which is computationally more tractable for the high-dimensional control
paths arising in our experiments.

One sweep visits all pairs $(n,c)$, with
\[
n=0,\ldots,N-1,
\qquad
c=1,\ldots,\frac{d(d+1)}{2},
\]
in lexicographic order. At a given pair $(n,c)$, all other coordinates
are kept fixed and $u_{n,c}$ is replaced by the element of
$\mathcal G_\delta$ that gives the largest value of the pathwise dual
objective in~\eqref{eq:hlr-discrete-dual}. This update is performed
independently for each Monte Carlo path $m$. Since each coordinate update
is an exact maximization over the corresponding one-dimensional mesh,
conditional on all the other coordinates, the pathwise objective is
nondecreasing along the coordinate updates.

This procedure belongs to the broad class of coordinate optimization
algorithms~\cite{Wri15}. If the coordinate sweeps are iterated until no
strict improvement occurs, and the current coordinate value is retained
in case of a tie, then the finite-mesh algorithm terminates after finitely
many strict improvements. The terminal control sequence is a
coordinate-wise maximizer of the discretized pathwise objective: no
change of a single control coordinate can increase the objective while
all the other coordinates are kept fixed. This property does not imply
global optimality of the finite-mesh problem, since simultaneous changes
of several coordinates may still yield a larger value. Moreover, without
additional structural assumptions on the pathwise objective, it does not
provide a general convergence guarantee toward either a local or a global
maximizer of the underlying continuous problem as $\delta\to0$.

The implementation also uses a locality cache. When testing candidates for coordinate $(n,c)$, the prefix of the simulated path up to time $t_n$, as well as the cumulative martingale and penalty terms up to $t_n$, are identical across all candidates. Hence it is unnecessary to resimulate the full trajectory. We cache the current states and cumulative contributions, forward-simulate only the suffix from step $n$ onward for each candidate, and refresh the cache after all coordinates at time step $n$ have been updated. The coordinate search is initialized from the trained primal actors. More precisely, the squashed actor output is projected onto the mesh $\mathcal G_\delta$, and the resulting mesh-valued control sequence is used as the initial point. The dual search therefore starts near the control suggested by the lower-bound algorithm, and the coordinate ascent acts as a pathwise refinement.

There is, of course, a trade-off between this scalable mesh-coordinate
approach and a full direct-search solver. A continuous direct-search method
over the full pathwise control vector avoids the additional finite-mesh
restriction and therefore provides a useful numerical benchmark when it
remains computationally feasible. However, its cost becomes prohibitive
beyond small $N$ and low dimension. We therefore keep $\delta>0$ fixed and
choose it sufficiently small. In low-dimensional test cases where full direct search is still computationally feasible, we observed no material difference between the resulting dual values and those obtained with the proposed mesh-coordinate search. In the numerical experiments below, we use $\delta=0.1$ and perform three coordinate sweeps. Empirically, additional sweeps produced no
material change in the resulting dual estimator. Accordingly, the coordinate-wise optimality property stated above applies to the fully converged coordinate search and is not assumed for the control sequences used in the reported experiments. Since the inner pathwise supremum is itself approximated numerically, the resulting quantity is interpreted throughout as a dual estimate rather than as a certified upper bound.

\subsubsection{Experiments}

We now compute the discretized dual estimate~\eqref{eq:hlr-discrete-dual}. 
For each test case, \Cref{tab:dual-diagnostics} reports the actor price 
obtained under the deterministic policy, the dual estimate obtained with the 
Sobolev critic gradient, and the dual estimate obtained with the gradient of 
a critic trained only on value targets. This comparison isolates the role of 
gradient training: in our experiments, the actor prices obtained with 
classical and Sobolev critic training are nearly identical. As explained 
in~\Cref{subsec:sobolev}, these experiments are designed to assess the quality of the learned gradient used as martingale integrand and approximate superhedging strategy. The associated dual estimates can also be compared with the actor estimates
and the reference prices.

\begin{table}[H]
\centering
\caption{Dual diagnostics: actor estimates and dual estimates obtained with Sobolev and classical critic gradients.}
\label{tab:dual-diagnostics}
\begin{adjustbox}{max width=\textwidth}
\begin{tabular}{@{} l c c c c c @{}}
\toprule
\textbf{Option} & $d$ & Actor price & Dual Sobolev & Dual classical & Reference price \\
\midrule
\multirow{4}{*}{\textbf{Geo-outperformer}}
& 2 & $13.75 \pm 0.06$ & $14.49 \pm 0.05$ & $17.97 \pm 0.05$ & 13.75 \\
& 3 & $12.85 \pm 0.05$ & $13.72 \pm 0.06$ & $17.71 \pm 0.07$ & 12.96 \\
& 4 & $12.60 \pm 0.05$ & $13.52 \pm 0.07$ & $16.87 \pm 0.07$ & 12.73 \\
& 5 & $12.47 \pm 0.05$ & $13.55 \pm 0.08$ & $16.13 \pm 0.08$ & 12.64 \\
\midrule
\textbf{Outperformer spread}
& 2 & $12.79 \pm 0.02$ & $13.44 \pm 0.04$ & $14.25 \pm 0.04$ & 12.83 \\
\midrule
\textbf{Best-of butterfly}
& 2 & $6.63 \pm 0.02$ & $7.65 \pm 0.03$ & $8.49 \pm 0.03$ & 6.70 \\
\bottomrule
\end{tabular}
\end{adjustbox}
\end{table}

The results show that Sobolev Training substantially improves the dual estimate. 
Across all test cases, the dual estimates based on the Sobolev critic gradient 
are much closer to the reference prices than those obtained from the classical 
critic gradient. This supports the conclusion that a critic trained only on value targets may 
provide an accurate pointwise price approximation while producing spatial 
derivatives of insufficient quality for martingale correction and hedging 
purposes.

The relevant conclusion is not that the dual estimates are tight in all cases, 
but that Sobolev Training produces a substantially better martingale integrand 
and hence a more reliable approximate superhedging strategy. A representative 
convergence plot is shown in~\Cref{fig:outperf-CV-dual}. The finite-difference 
curve corresponds to the dual estimate obtained by using the gradient computed 
from the finite-difference scheme of~\cite{MF17}. It provides a high-accuracy 
benchmark for the martingale integrand in this two-dimensional test case. The 
Sobolev curve tracks this finite-difference benchmark closely, while the 
classical critic gradient leads to a substantially looser dual estimate. This 
indicates that the delta learned with Sobolev Training gives a good numerical 
approximation of the finite-difference benchmark delta.

\begin{figure}[H]
\centering
\includegraphics[width=1.0\textwidth]{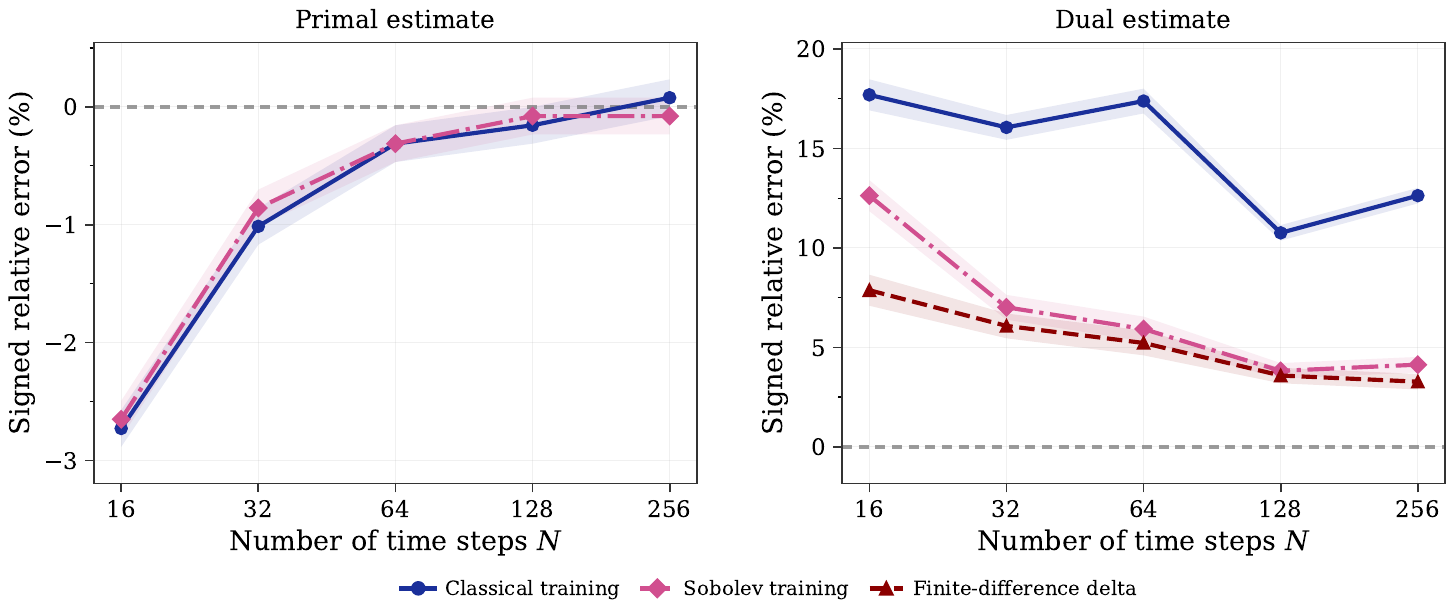}
\caption{Signed relative errors of the primal and dual estimates for the 
outperformer spread option. The left panel compares the actor estimates obtained 
with classical and Sobolev critic training. The right panel compares the dual 
estimates obtained with the corresponding critic gradients and with the 
finite-difference delta benchmark. Shaded bands represent $95\%$ Monte Carlo 
confidence intervals.}
\label{fig:outperf-CV-dual}
\end{figure}

The moderate tightness of the dual estimates should also be interpreted in
light of time discretization. The dual construction relies on a
piecewise-constant policy approximation, for which
\cite[Theorem~2.1]{JPR19} gives an order $(T/N)^{1/4}$ under the assumptions
stated above. We therefore use $N^{-1/4}$ as a reference slope rather than as
a claimed convergence rate for the fully implemented dual estimator.
As illustrated in Appendix~\ref{app:dual-rate}, the finite-difference and
Sobolev dual errors decrease broadly in line with this reference slope, while
the classical-gradient dual estimate does not show a comparable trend.
Thus, the remaining dual gap appears to reflect both the quality of the
martingale integrand and the discretization inherent in the dual construction.

\bibliography{references}

\appendix

\section{Proofs}\label{app:proofs}

\begin{proof}[Proof of Proposition~\ref{prop:uvm-basis-independence}]
Fix $(t,x)\in[0,T)\times(0,+\infty)^d$. Let $\nu_t$ be a reference
probability system and $\alpha\in\mathcal A_{\nu_t}$. Since every component of
$X^{\alpha,t,x}$ remains strictly positive, define componentwise
\[
Y_s^{\alpha,t,x}\coloneqq\log X_s^{\alpha,t,x},
\qquad s\in[t,T].
\]
It\^o's formula gives
\begin{equation}\label{eq:proof-log-dynamics}
dY_s^{\alpha,t,x}
=
b(\alpha_s)\,ds+\alpha_s\,dW_s,
\qquad
Y_t^{\alpha,t,x}=\log x,
\end{equation}
where
\[
b^i(a)\coloneqq r-\frac12(aa^\top)^{ii},
\qquad i=1,\ldots,d.
\]
Thus, in logarithmic coordinates, the control set is the compact set
$A\subset\mathbb R^{d\times d}$, the drift is $b(a)$, and the diffusion
coefficient is $a$.

We first consider a bounded uniformly continuous terminal reward
$\varphi:\mathbb R^d\to\mathbb R$. The coefficients in
\eqref{eq:proof-log-dynamics} are bounded, continuous in the control variable,
and independent of time and state, so the assumptions in
\cite[Chapter~IV, (6.1)]{FS06} are immediate. The running cost is zero and
therefore satisfies~\cite[Chapter~IV, (6.3)]{FS06}; compactness of $A$ also
makes the integrability requirement in the admissible-control class of
\cite[Section~IV.2]{FS06} automatic. Applying
\cite[Corollary~IV.7.1]{FS06}, based on~\cite[Theorem~IV.7.1]{FS06}, to the
terminal cost $-\varphi$ and changing signs gives, for any two reference
probability systems $\nu_t$ and $\nu_t'$,
\begin{equation}\label{eq:proof-FS-buc}
\sup_{\alpha\in\mathcal A_{\nu_t}}
\mathbb E^{\nu_t}\!\left[\varphi\bigl(Y_T^{\alpha,t,x}\bigr)\right]
=
\sup_{\alpha\in\mathcal A_{\nu_t'}}
\mathbb E^{\nu_t'}\!\left[\varphi\bigl(Y_T^{\alpha,t,x}\bigr)\right].
\end{equation}

We now extend this identity to the payoff in~\Hyp{H1}. Set
\[
h(y)\coloneqq g(e^{y^1},\ldots,e^{y^d}),
\qquad y\in\mathbb R^d,
\]
and, for $R\geq1$, choose $\chi_R\in C_c(\mathbb R^d)$ such that
$0\leq\chi_R\leq1$, $\chi_R=1$ on $\{|y|\leq R\}$ and
$\chi_R=0$ on $\{|y|\geq R+1\}$. Define $h_R\coloneqq\chi_Rh$. Then
$h_R\in C_c(\mathbb R^d)$ and hence is bounded and uniformly continuous.

By~\Hyp{H1} and~\eqref{eq:moment-estimate}, and by integrating
\eqref{eq:proof-log-dynamics} and using compactness of $A$, there is a constant
$C_{t,x}<\infty$, independent of the reference probability system and the
admissible control, such that
\[
\mathbb E^{\nu_t}\!\left[\left|h\bigl(Y_T^{\alpha,t,x}\bigr)\right|^2\right]
\leq C_{t,x},
\qquad
\mathbb E^{\nu_t}\!\left[\left|Y_T^{\alpha,t,x}\right|^2\right]
\leq C_{t,x}.
\]
Consequently, Markov's inequality and Cauchy-Schwarz yield, after enlarging
$C_{t,x}$ if necessary,
\begin{equation}\label{eq:proof-uniform-truncation}
\sup_{\nu_t}\sup_{\alpha\in\mathcal A_{\nu_t}}
\mathbb E^{\nu_t}\!\left[
\left|h\bigl(Y_T^{\alpha,t,x}\bigr)-h_R\bigl(Y_T^{\alpha,t,x}\bigr)\right|
\right]
\leq \frac{C_{t,x}}{R}.
\end{equation}
Indeed,
$|h-h_R|(y)\leq |h(y)|\mathbf 1_{\{|y|>R\}}$, while the second moment bound on
$Y_T^{\alpha,t,x}$ gives a tail probability of order $R^{-2}$ uniformly in
$\nu_t$ and $\alpha$.

Since $h_R$ is bounded and uniformly continuous,
\eqref{eq:proof-FS-buc} applies to $h_R$. Using the elementary inequality
$|\sup F-\sup G|\leq\sup|F-G|$ on each reference system and
\eqref{eq:proof-uniform-truncation}, we therefore obtain
\[
\left|V_{\nu_t}(t,x)-V_{\nu_t'}(t,x)\right|
\leq
2\,\mathrm e^{-r(T-t)}\frac{C_{t,x}}{R}.
\]
Letting $R\to\infty$ gives
$V_{\nu_t}(t,x)=V_{\nu_t'}(t,x)$. Since $\nu_t'$ is arbitrary, taking the
supremum over reference probability systems in~\eqref{eq:uvm-weak-value} yields
$V(t,x)=V_{\nu_t}(t,x)$.
\end{proof}

\begin{proof}[Proof of \Cref{lem:dDPP}]
We apply the dynamic programming theorem for finite-horizon Markov control processes (Theorem~3.2.1 in~\cite{HLL96}, Chapter~3). Since our problem involves the maximization of a discounted terminal reward rather than the minimization of a running cost, we use the variants described in Section~3.4 of~\cite{HLL96}: the maximization form (equations~(3.4.13)-(3.4.14)) and the discounted form (equations~(3.4.7)-(3.4.8)).

To apply Theorem~3.2.1 in~\cite{HLL96}, it suffices to verify the measurable selection condition (Assumption~3.3.1 in~\cite{HLL96}). We check that Condition~3.3.3 in~\cite{HLL96} is satisfied, which implies the measurable selection condition by Theorem~3.3.5(i) therein for any nonnegative measurable function. If $g$ is not nonnegative, choose $C\in\mathbb R$ such that $g+C\geq0$ (possible under~\Hyp{H4}) and 
apply the theorem to the nonnegative terminal reward $g+C$. The corresponding 
value functions are shifted by the deterministic quantity 
$\mathrm e^{-r(T-t_n)}C$, so the maximizing selectors and the dynamic 
programming recursion for $g$ are unchanged after subtracting this shift.

\begin{enumerate}[label=(\alph*)]
    \item The action set $A$ is compact and does not depend on the state, so the multifunction $x\mapsto A$ is trivially upper semicontinuous.
    \item There is no running cost in our problem, so condition~(b) is vacuous.
    \item We verify the weak continuity of the transition law. Let $\phi:(0,+\infty)^d\to\mathbb{R}$ be a continuous bounded function and let $\xi\sim\mathcal{N}(0,I_d)$. From the definition~\eqref{eq:log-euler} of $F$, the map
    \[
    (0,+\infty)^d\times A\ni(x,a)\mapsto F(x,a,\xi)\in(0,+\infty)^d
    \]
    is continuous. Therefore $(x,a)\mapsto\phi\big(F(x,a,\xi)\big)$ is continuous and bounded by composition. By Lebesgue's dominated convergence theorem, the map
    \[
    (0,+\infty)^d\times A\ni(x,a)\mapsto \mathbb{E}\left[\phi\big(F(x,a,\xi)\big)\right]\in\mathbb{R}
    \]
    is continuous and bounded, i.e., the transition law is weakly continuous.
\end{enumerate}

Since Condition~3.3.3 in~\cite{HLL96} is satisfied, Theorem~3.3.5(i) therein guarantees the measurable selection condition for any nonnegative measurable function. The value functions $(\mathcal{V}_n^\star)_{0\leq n\leq N}$ are bounded below: this holds at $n=N$ by~\Hyp{H4}, and propagates backward since the discount factor $\mathrm{e}^{-r T/N}$ is positive. Therefore, the dynamic programming theorem (Theorem~3.2.1 in~\cite{HLL96}) applies at each backward step, yielding both the recursion~\eqref{eq:dDPP} and the existence of measurable maximizers $(\alpha_n^\star)_{0\leq n\leq N-1}$.
\end{proof}

\begin{proof}[Proof of~\Cref{lem:convergence-rate}]
Since $\Hyp{H2prime}$ implies $\Hyp{H1}$, Proposition~\ref{prop:uvm-basis-independence}
allows us, for the purpose of this proof, to evaluate $V(t_n,x)$ on the
canonical Wiener space over $[t_n,T]$. Let $h=T/N$, and let $v_h(t_n,x)$
denote the value obtained there by restricting the admissible controls to
controls that are constant on each interval $[t_k,t_{k+1})$,
$n\leq k\leq N-1$. Thus
\[
\alpha_s=a_k,\qquad s\in[t_k,t_{k+1}),
\]
where $a_k$ is $A$-valued and measurable with respect to the canonical
filtration at time $t_k$.

Under~\Hyp{H2prime}, $g$ admits a globally Lipschitz extension to
$\mathbb R^d$ with the same Lipschitz constant. Since the controlled state
started in $(0,+\infty)^d$ remains in $(0,+\infty)^d$, this extension does
not alter the value functions. Moreover, the coefficients of
\eqref{eq:UVM-multi-dynamics} are Lipschitz with linear growth in the state,
uniformly in $a\in A$. The piecewise constant policy estimate of
\cite[Theorem~2.1]{JPR19}, together with the linear-growth extension
discussed in~\cite[Section~2, after Eq.~(2.4)]{PR20}, therefore yields constants
$C>0$ and $q\geq0$, independent of $N$, $n$, and $x$, such that
\begin{equation}\label{eq:proof-PR-bound}
0
\leq
V(t_n,x)-v_h(t_n,x)
\leq
C(1+|x|^q)h^{1/4}.
\end{equation}

It remains to identify $v_h(t_n,x)$ with $\mathcal V_n^\star(x)$.
By~\cite[Eq.~(2.4)]{PR20}, the value function $v_h$ satisfies the
backward recursion
\[
v_h(T,x)=g(x),
\qquad
v_h(t_k,x)
=
\sup_{a\in A}
\mathbb E\big[e^{-rh}\,v_h\big(t_{k+1},X_{t_{k+1}}^{a}\big)\big],
\qquad 0\leq k\leq N-1,
\]
where, given $X_{t_k}=x$ and constant control $a\in A$
on $[t_k,t_{k+1})$, the SDE~\eqref{eq:UVM-multi-dynamics} reduces
on this interval to a geometric Brownian motion with constant
coefficients, whose strong solution at time $t_{k+1}$ is
\[
X_{t_{k+1}}^a
=
x\odot
\exp\Big\{
\big(r\mathbf 1_d-\tfrac12\operatorname{diag}(aa^\top)\big)h
+a\sqrt h\,\xi_k
\Big\}
=
F(x,a,\xi_k),
\]
with $\xi_k\sim\mathcal N(0,I_d)$. Hence
\[
v_h(t_k,x)
=
\sup_{a\in A}\mathbb E\big[e^{-rh}\,v_h\big(t_{k+1},F(x,a,\xi_k)\big)\big],
\qquad n\leq k\leq N-1,
\]
which is exactly the discrete DPP recursion~\eqref{eq:dDPP} satisfied
by $(\mathcal V_n^\star)$, with the same terminal condition. In particular,
the DPP shows that allowing randomized Markov policies in the discrete problem
does not increase its optimal value. By backward induction,
$v_h(t_k,x)=\mathcal V_k^\star(x)$ for all
$n\leq k\leq N$ and $x\in(0,+\infty)^d$. In particular,
$v_h(t_n,x)=\mathcal V_n^\star(x)$, and combining this with
\eqref{eq:proof-PR-bound} yields
\[
0
\leq
V(t_n,x)-\mathcal V_n^\star(x)
\leq
C(1+|x|^q)h^{1/4},
\]
which proves~\eqref{eq:cv-rate}.
\end{proof}

\begin{proof}[Proof of the score function identity~\eqref{eq:reinforce}]
Fix $x\in(0,+\infty)^d$ and $\theta\in\Theta$. Under~\Hyp{H5}, the policy $\pi_\theta(\cdot\mid x)$ admits a density $p_\theta(\cdot\mid x)$ with respect to a measure $\kappa$ on the action space. Therefore,
\[
\mathcal{J}_n(x;\theta)
=
\int_A
\mathbb{E}\left[
\mathrm{e}^{-r\frac{T}{N}}
\mathcal{V}^\star_{n+1}\big(F(x,a,\xi)\big)
\right]
p_\theta(a\mid x)\,\kappa(da).
\]
Under the integrability assumptions induced by the polynomial-growth payoff and the standard moment bounds for the log-Euler scheme, the above quantity is finite. Under standard regularity conditions on the parameterization $\theta\mapsto p_\theta(a\mid x)$ (satisfied by the policy classes of \Cref{subsec:policies}), differentiation under the integral sign yields
\[
\nabla_\theta\,\mathcal{J}_n(x;\theta)
=
\int_A
\mathbb{E}\left[
\mathrm{e}^{-r\frac{T}{N}}
\mathcal{V}^\star_{n+1}\big(F(x,a,\xi)\big)
\right]
\nabla_\theta\, p_\theta(a\mid x)\,\kappa(da).
\]
Using the log-derivative identity $\nabla_\theta\, p_\theta(a\mid x) = p_\theta(a\mid x)\,\nabla_\theta\log p_\theta(a\mid x)$, we obtain
\[
\nabla_\theta\,\mathcal{J}_n(x;\theta)
=
\int_A
\mathbb{E}\left[
\mathrm{e}^{-r\frac{T}{N}}
\mathcal{V}^\star_{n+1}\big(F(x,a,\xi)\big)
\right]
\nabla_\theta\log p_\theta(a\mid x)\,
p_\theta(a\mid x)\,\kappa(da),
\]
which is precisely~\eqref{eq:reinforce}.

The control variate identity $\mathbb{E}^{a\sim\pi_\theta(\cdot\mid x)}[b(x)\,\nabla_\theta\log p_\theta(a\mid x)] = 0$ for any measurable $b:(0,+\infty)^d\to\mathbb{R}$ follows from the same argument applied to the constant function (in $a$) equal to $b(x)$: since $\int_A p_\theta(a\mid x)\,\kappa(da) = 1$ for all $\theta$, differentiating both sides gives $\int_A \nabla_\theta\log p_\theta(a\mid x)\, p_\theta(a\mid x)\,\kappa(da) = 0$.
\end{proof}

\section{Additional numerical experiments}\label{app:numerics}

\subsection{Sensitivity to the penalty weight}
\label{app:penalty-sensitivity}

We investigate the effect of the penalty weight $\beta$ on 
the geo-outperformer option in dimensions $d=3,4,5$. \Cref{fig:beta-sensitivity} displays, for each dimension, the signed relative error of the actor price with respect to the reference price, and the relative price impact of correlation bound violations, defined as 
\[
\frac{\text{unclamped price}-\text{clamped price}}{\text{reference price}},
\]
where the clamped variant projects 
all pairwise correlations onto 
$[\rho^{ij}_{\min},\rho^{ij}_{\max}]$ at each time step. For small values of $\beta$ ($\leq 1$), the penalty is insufficient: the actor price is significantly inflated by correlation violations, as evidenced by the large gap 
between the unclamped and clamped prices. For $\beta=10$, the price impact becomes negligible and the clamped price is close to the reference. For larger values ($\beta\geq 10^2$), the observed constraint violations become negligible, but the actor price deteriorates as the penalty term competes with the PPO objective during training. 
The value $\beta=10$ used throughout the paper strikes the best observed balance between constraint enforcement and optimization quality.

\begin{figure}[H]
\centering
\includegraphics[width=1.0\textwidth]{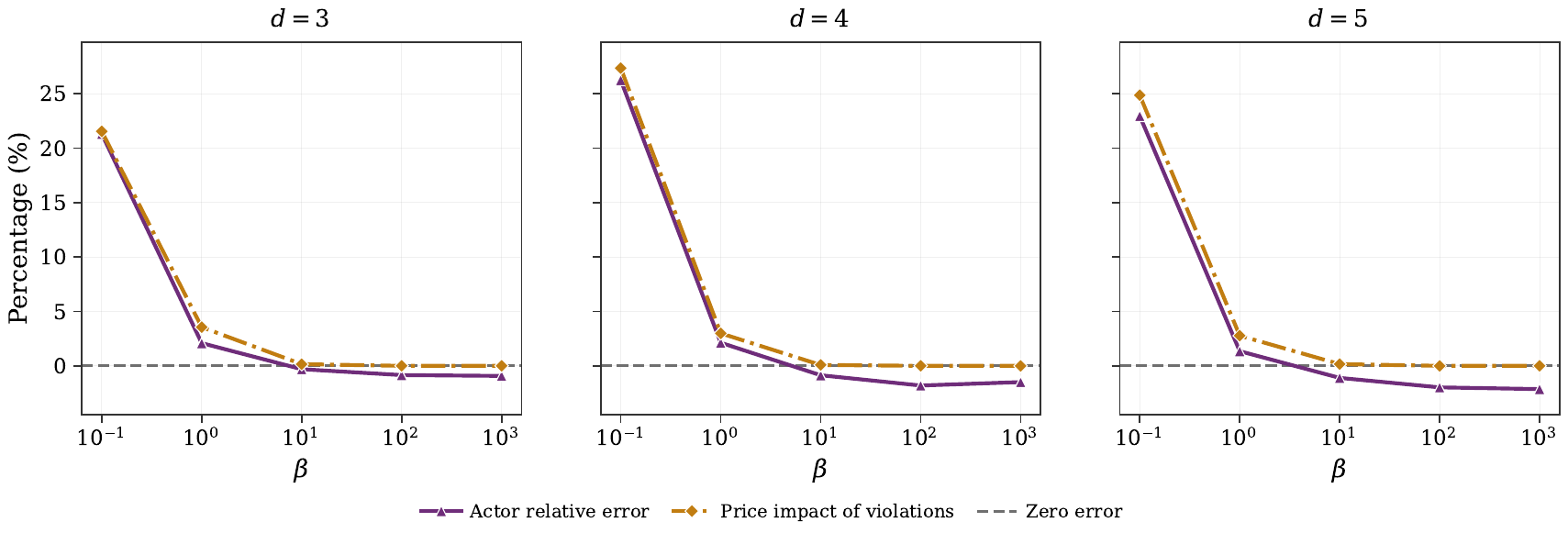}
\caption{Sensitivity to the penalty weight $\beta$ for the geo-outperformer 
option. For each dimension, the signed actor relative error and the relative price 
impact of correlation-bound violations are shown as percentages.}
\label{fig:beta-sensitivity}
\end{figure}

We emphasize that the clamped variant used in this analysis is a diagnostic tool, not a viable alternative to the penalty approach. Projecting pairwise correlations 
onto $[\rho^{ij}_{\min},\rho^{ij}_{\max}]$ does not guarantee that the resulting matrix remains positive semidefinite. In the present example, the violations are small enough that the projected matrix happens to remain 
valid, but this cannot be expected in general, especially for tighter bounds or higher dimensions. The penalty mechanism, by contrast, penalizes violations of the bounds while preserving the positive semidefiniteness guaranteed by the $C$-vine construction.

\subsection{Additional convergence plots}
\label{app:convergence-plots}

\Cref{fig:conv-app-1,fig:conv-app-2} display the convergence of the actor and critic price estimates with respect to $N$ for best-of butterfly and the geo-outperformer ($d=3$) options. The same qualitative behavior as in~\Cref{subsec:uc} is observed: the actor 
price approaches the reference from below, while the critic price converges in 
absolute error but with variable sign across runs.

\begin{figure}[H]
\centering
\includegraphics[width=1.0\textwidth]{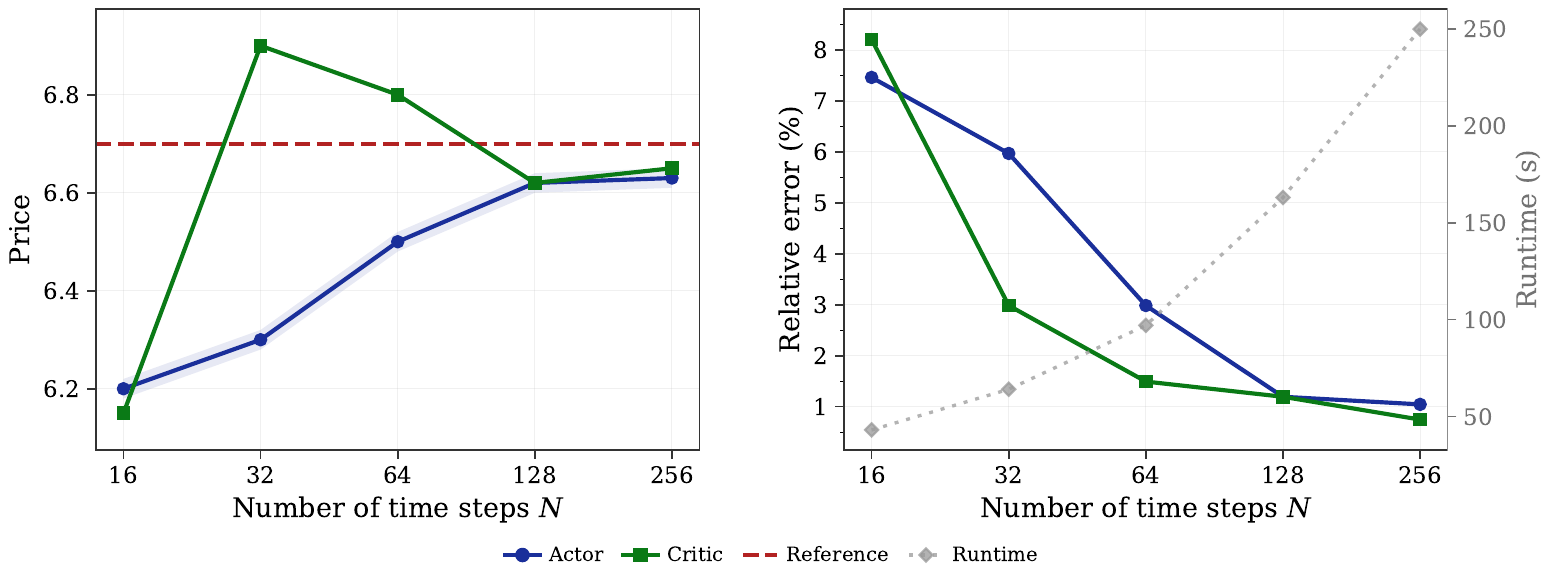}
\caption{Convergence with respect to the number of time steps $N$ for the 
best-of butterfly option. The left panel reports actor and critic prices, 
together with the reference price; the right panel reports absolute relative 
errors.}
\label{fig:conv-app-1}
\end{figure}

\begin{figure}[H]
\centering
\includegraphics[width=1.0\textwidth]{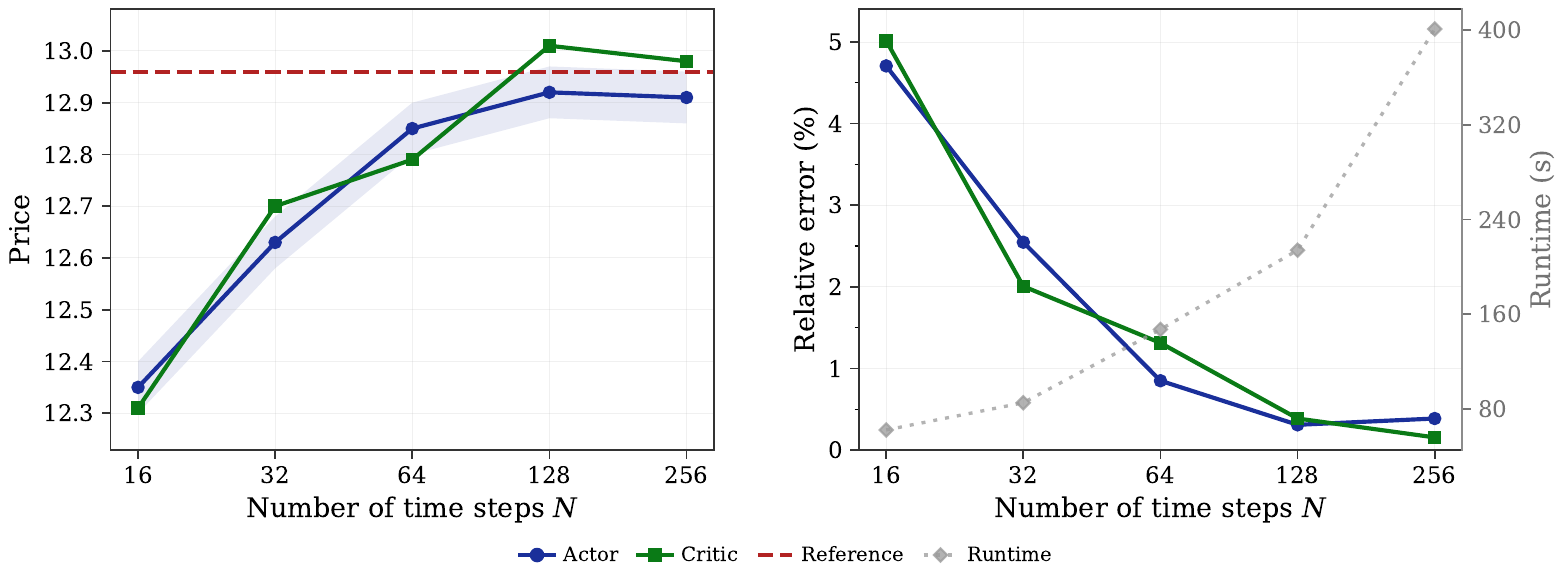}
\caption{Convergence with respect to the number of time steps $N$ for the 
geo-outperformer option ($d=3$, uncertain correlation). The left panel reports 
actor and critic prices, together with the reference price; the right panel 
reports absolute relative errors.}
\label{fig:conv-app-2}
\end{figure}

\subsection{Empirical convergence rate study}\label{app:dual-rate}

\Cref{fig:outperf-dual-loglog} reports a log-log convergence diagnostic for the
outperformer spread test case. The plotted quantity is the absolute relative
error with respect to the reference price,
\[
\frac{|\widehat V_N - V_{\mathrm{ref}}|}{|V_{\mathrm{ref}}|},
\]
as a function of the number of time steps $N$. The dashed gray line is
proportional to $N^{-1/4}$ and is included only as a visual reference for the piecewise-constant policy approximation rate of~\cite{JPR19}.

The comparison between the two panels highlights the different numerical 
behavior of the primal and dual estimators. The primal actor error decreases faster than the $N^{-1/4}$ reference slope over this test range, 
whereas the Sobolev and finite-difference dual errors exhibit a slower decay, closer to the $N^{-1/4}$ reference slope. This 
difference is consistent with the fact that the primal estimate evaluates a 
fixed learned policy, while the dual estimate involves a pathwise maximization 
and a martingale correction accumulated over the time grid.

For the dual estimates, the finite-difference and Sobolev-gradient martingale 
integrands show a decreasing trend broadly consistent with the $N^{-1/4}$ 
guide, whereas the dual estimate obtained from the classical critic gradient 
remains larger and does not display a comparable convergence pattern. This 
supports the interpretation given in~\Cref{sec:dual}: Sobolev Training improves 
the learned martingale integrand, while the remaining gap is partly explained 
by the time-discretization effects inherent in the dual construction.

\begin{figure}[H]
\centering
\includegraphics[width=1.0\textwidth]{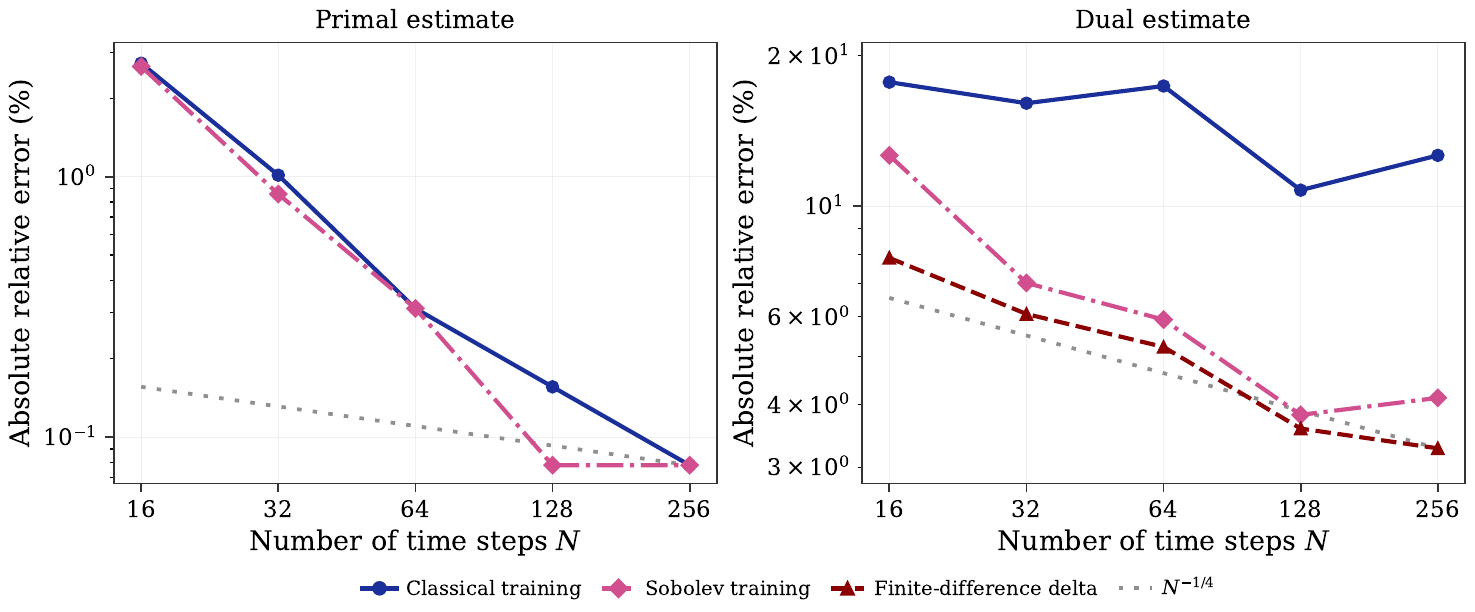}
\caption{Log-log convergence diagnostic for the outperformer spread option. The
vertical axis reports absolute relative errors with respect to the reference
price. The gray dashed line is proportional to $N^{-1/4}$.}
\label{fig:outperf-dual-loglog}
\end{figure}

Finally, \Cref{fig:geo-dual-loglog} reports a log-log convergence diagnostic for the geo-outperformer option, with $d=4$.

\begin{figure}[H]
\centering
\includegraphics[width=1.0\textwidth]{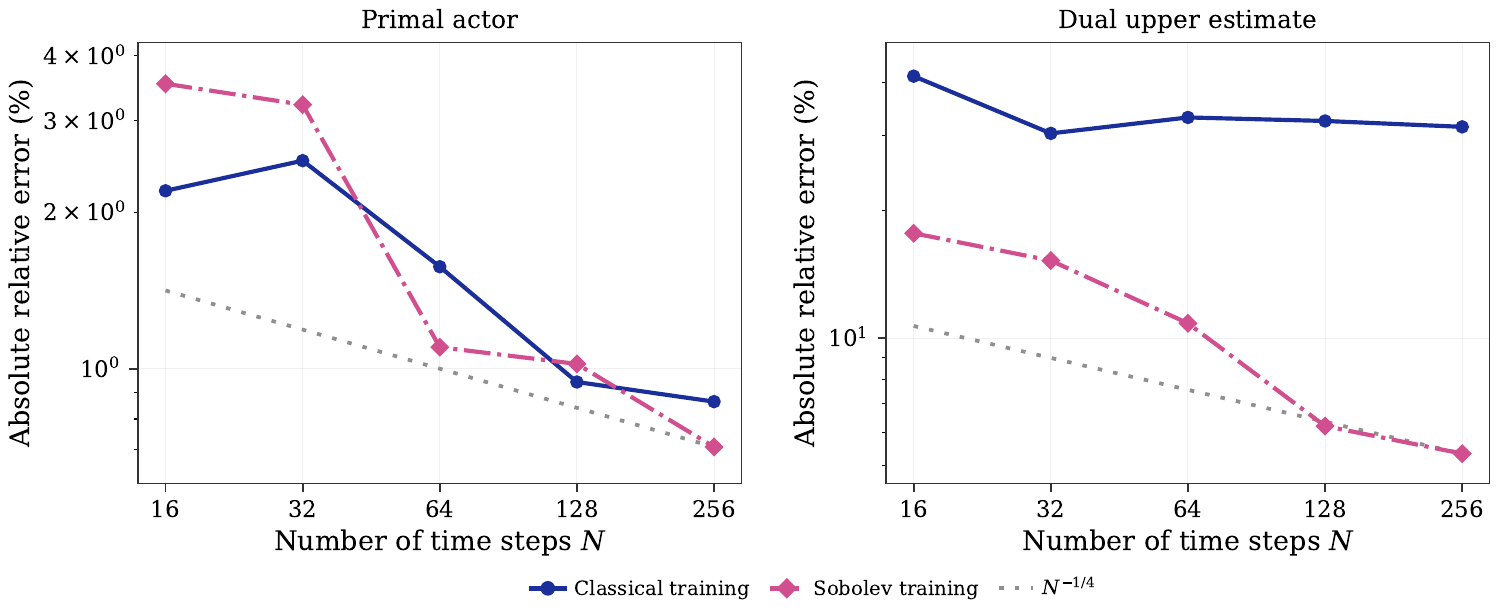}
\caption{Log-log convergence diagnostic for the geo-outperformer ($d=4$) option. The
vertical axis reports absolute relative errors with respect to the reference
price. The gray dashed line is proportional to $N^{-1/4}$.}
\label{fig:geo-dual-loglog}
\end{figure}

\end{document}